\documentclass[acmsmall,screen]{acmart}

\setcopyright{rightsretained}
\acmPrice{}
\acmDOI{10.1145/3371082}
\acmYear{2020}
\copyrightyear{2020}
\acmJournal{PACMPL}
\acmVolume{4}
\acmNumber{POPL}
\acmArticle{14}
\acmMonth{1}
\startPage{1}

\usepackage{booktabs}   %
\usepackage{subcaption} %

\usepackage{amssymb}
\usepackage{amsmath}

\usepackage{mathpartir}

\usepackage{booktabs}

\usepackage{ifpdf}

\usepackage{graphicx}

\usepackage{enumitem}

\usepackage{mathtools}

\usepackage{epstopdf}

\usepackage[labelformat=simple]{subcaption} %

\usepackage{listings}

\usepackage{scalerel}

\usepackage{xcolor}

\usepackage{tikz}
\usepackage{mdframed}
\usepackage{array}
\usepackage{multicol}
\usepackage{mwe}
\usepackage{wrapfig}
\usetikzlibrary{shapes,arrows,automata,positioning,cd,calc}
\usepgflibrary{plotmarks}
\tikzset{
  cfgedge/.style   = {black, ->, >=stealth},
  cfgedgestar/.style   = {black, ->, >=stealth, shorten >=.17em,
                          to path={-- node[inner sep=0pt,at end,sloped] {${}^*$} (\tikztotarget) \tikztonodes}},
  cfgedgeplus/.style   = {black, ->, >=stealth, shorten >=.17em,
                          to path={-- node[inner sep=0pt,at end,sloped] {${}^+$} (\tikztotarget) \tikztonodes}},
  forward/.style = { blue, ->, >=angle 45},
  forwardstar/.style = {blue, ->, >=angle 45, shorten >=.17em,
             to path={-- node[inner sep=0pt,at end,sloped] {${}^*$} (\tikztotarget) \tikztonodes}},
  forwardplus/.style = {blue, ->, >=angle 45, shorten >=.17em,
             to path={-- node[inner sep=0pt,at end,sloped] {${}^+$} (\tikztotarget) \tikztonodes}},
  backward/.style = { red, densely dashed, ->, >=latex' },
  backwardstar/.style = { red, densely dashed, ->, >=latex', shorten >=.17em,
                      to path={-- node[inner sep=0pt,at end,sloped] {${}^*$} (\tikztotarget) \tikztonodes}},
}

\newcommand{\cfgarrow}{\mathbin{\tikz[baseline]\path (0,0) edge[cfgedge, yshift=0.6ex] (.9em,0);}}
\newcommand{\cfgarrowstar}{\mathbin{\tikz[baseline]\path (0,0) edge[cfgedgestar, yshift=0.6ex] (.9em,0);}}
\newcommand{\cfgarrowplus}{\mathbin{\tikz[baseline]\path (0,0) edge[cfgedgeplus, yshift=0.6ex] (.9em,0);}}

\newcommand{\forwardarrow}{\mathbin{\tikz[baseline]\path (0,0) edge[forward, yshift=0.6ex] (1em,0);}}
\newcommand{\forwardarrowstar}{\mathbin{\tikz[baseline]\path (0,0) edge[forwardstar, yshift=0.6ex] (1em,0);}}
\newcommand{\forwardarrowplus}{\mathbin{\tikz[baseline]\path (0,0) edge[forwardplus, yshift=0.6ex] (1em,0);}}
\newcommand{\subforwardarrowstar}{\mathbin{\tikz[baseline]\path (0,0) edge[forwardstar, yshift=-0.2ex] (.85em,0);}}

\newcommand{\backwardarrow}{\mathbin{\tikz[baseline]\path (0,0) edge[backward, yshift=0.6ex] (.95em,0);}}

\usetikzlibrary{fit,calc}
\newcommand*{\tikzmk}[1]{\tikz[remember picture,overlay,] \node (#1) {};\ignorespaces}
\newcommand{\boxit}[2]{\tikz[remember picture,overlay]{\node[yshift=3pt,xshift=3pt,fill=#1,opacity=.15,fit={(A)($(B)+(#2\linewidth,.8\baselineskip)$)}] {};}\ignorespaces}
\colorlet{pink}{red!40}
\colorlet{cyan}{cyan!60}
\colorlet{gray}{gray!60}

\usepackage{prettyref}
\newcommand{\pref}{\prettyref}
\newrefformat{thm}{Theorem~\ref{#1}}
\newrefformat{lem}{Lemma~\ref{#1}}
\newrefformat{cor}{Corollary~\ref{#1}}
\newrefformat{cha}{Chapter~\ref{#1}}
\newrefformat{sec}{\S\!~\ref{#1}}
\newrefformat{tab}{Table~\ref{#1}}
\newrefformat{fig}{Figure~\ref{#1}}
\newrefformat{alg}{Algorithm~\ref{#1}}
\newrefformat{exa}{Example~\ref{#1}}
\newrefformat{def}{Definition~\ref{#1}}
\newrefformat{li}{Line~\ref{#1}}
\newrefformat{eq}{Eq.~\ref{#1}}
\newrefformat{case}{Case~\ref{#1}}

\newlist{ecomponents}{enumerate}{1}
\setlist[ecomponents,1]{label={\bfseries C\arabic*},align=left}

\newlist{casesp}{enumerate}{3} %
\setlist[casesp]{align=left, %
                 listparindent=\parindent, %
                 parsep=\parskip, %
                 font=\normalfont\bfseries, %
                 leftmargin=0pt, %
                 labelwidth=0pt, %
                 itemindent=.4em,labelsep=.4em, %
                 partopsep=0pt, %
                 }
\setlist[casesp,1]{label=Case~\arabic*,ref=\arabic*,leftmargin=0ex}
\setlist[casesp,2]{label=Case~\thecasespi.\alph*,ref=\thecasespi.\alph*,leftmargin=2ex}
\setlist[casesp,3]{label=Case~\thecasespii.\roman*,ref=\thecasespii.\roman*,leftmargin=2ex}

\usepackage[ruled,vlined,linesnumbered]{algorithm2e}

\definecolor{lapislazuli}{rgb}{0.15, 0.38, 0.61}
\SetAlgoNlRelativeSize{-1} 

\SetCommentSty{CommentFont}
\SetKwComment{Comment}{$\triangleright$ }{}

\definecolor{light-gray}{gray}{0.9}
\definecolor{light-pink}{rgb}{0.858, 0.188, 0.478}

\definecolor{maroon}{rgb}{0.5, 0.0, 0.0}

\newcommand{\etch}{{\color{maroon}h}}

\newcommand{\Omit}[1]{}

\newcommand{\AbsDomain}{\mathcal{A}}

\newcommand{\xqed}[1]{%
 \leavevmode\unskip\penalty9999 \hbox{}\nobreak\hfill
  \quad\hbox{\ensuremath{#1}}}

\newcommand*{\qef}{\xqed{\scriptstyle{\blacksquare}}}

\newcommand{\subsubsubsection}[1]{\smallskip\noindent\textbf{\emph{#1}}\enspace}
\newcommand{\RQ}[1]{\textbf{RQ#1}}

\usepackage[xcolor]{changebar}
\cbcolor{red}

\newcommand{\InitRule}{\textsc{Init}}

\newcommand{\NERule}{\textsc{NonExit}}
\newcommand{\ENCRule}{\textsc{CompNotStabilized}}
\newcommand{\ECRule}{\textsc{CompStabilized}}

\newcommand{\nesting}{\mathsf{N}}
\newcommand{\relation}{\mathsf{R}}
\newcommand{\widen}{\triangledown}

\newcommand{\wpo}{\mathcal{W}}
\newcommand{\hpo}{\mathcal{H}}

\newcommand{\Value}{\mathcal{X}}
\newcommand{\Map}{\mathcal{N}}
\newcommand{\components}{\mathcal{C}}
\newcommand{\maximalcomponents}{\components^{0}}

\newcommand{\forward}{scheduling constraint}
\newcommand{\backward}{stabilization constraint}

\def\llceil{\lceil\kern-3pt\lceil}
\def\rrceil{\rceil\kern-3pt\rceil}
\def\llfloor{\lfloor\kern-3pt\lfloor}
\def\rrfloor{\rfloor\kern-3pt\rfloor}

\newcommand{\postset}[1]{\llfloor {#1}\rceil}
\newcommand{\preset}[1]{\lfloor{#1}\rrceil}
\newcommand{\component}[2]{\llfloor {#1}, {#2}\rrceil}

\newcommand{\eqdef}{\buildrel \mbox{\tiny\textrm{def}} \over =}

\newcommand{\pikosname}{\textsc{Pikos}}
\newcommand{\pikos}[1]{{\textsc{Pikos\textlangle{}}$#1$\textsc{\textrangle{}}}}

\newcommand{\ikos}{\textsc{IKOS}}

\newcommand{\IntegerSet}{\mathbb{Z}}

\newcommand{\zerodisplayskips}{%
  \setlength{\abovedisplayskip}{0pt}%
  \setlength{\belowdisplayskip}{0pt}%
  \setlength{\abovedisplayshortskip}{0pt}%
  \setlength{\belowdisplayshortskip}{0pt}}

\newcommand{\totalsvc}{2701}%
\newcommand{\totaloss}{1618}%
\newcommand{\totalboth}{4319}%
\newcommand{\totalverylong}{130}%
\newcommand{\timelimit}{4 hours}%
\newcommand{\totalbench}{1017}%
\newcommand{\totalbenchsvc}{518}%
\newcommand{\totalbenchoss}{499}%

\newcommand{\maxspeed}{3.63x}%
\newcommand{\ameanspeed}{2.06x}%
\newcommand{\gmeanspeed}{1.95x}%
\newcommand{\hmeanspeed}{1.84x}%
\newcommand{\totalspeed}{2.16x}%

\newcommand{\abovetwonum}{533}%
\newcommand{\abovetwo}{52.4\%}%
\newcommand{\abovethreenum}{106}%
\newcommand{\abovethree}{10.4\%}

\newcommand{\ameanspeedl}{2.38x}%
\newcommand{\gmeanspeedl}{2.29x}
\newcommand{\hmeanspeedl}{2.18x}

\newcommand{\maxfreqai}{1.25x-1.50x}
\newcommand{\maxfreqa}{52}
\newcommand{\maxfreqbi}{1.50x-1.75x}
\newcommand{\maxfreqb}{45}
\newcommand{\maxfreqci}{3.00x-3.25x}
\newcommand{\maxfreqc}{38}
\newcommand{\maxfreqdi}{3.00x-3.25x}
\newcommand{\maxfreqd}{50}

\newcommand{\bestscaleall}{1.77x to 3.63x, 5.07x, 6.57x, 9.02x, and 10.97x}
\newcommand{\bestscale}{1.77x to 3.63x, 5.07x, and 6.57x}
\newcommand{\ascale}{1.48x, 2.06x, 2.26x, and 2.46x}
\newcommand{\gscale}{1.46x, 1.95x, 2.07x, and 2.20x}
\newcommand{\hscale}{1.44x, 1.84x, 1.88x, and 1.98x}

\newcommand{\bcoeff}{0.79}
\newcommand{\bspeed}{10.97x}

\newcommand{\coefff}{124} %

\DeclareMathOperator*{\argmin}{arg\,min}

\begin{document}

\title{Deterministic Parallel Fixpoint Computation}

\author{Sung Kook Kim}
\affiliation{
  \department{Computer Science}
  \institution{University of California, Davis}
  \city{Davis}
  \state{California}
  \postcode{95616}
  \country{U.S.A.}
}
\email{sklkim@ucdavis.edu}

\author{Arnaud J.\ Venet}
\affiliation{
  \institution{Facebook, Inc.}
  \city{Menlo Park}
  \state{California}
  \postcode{94025}
  \country{U.S.A.}
}
\email{ajv@fb.com}

\author{Aditya V.\ Thakur}
\affiliation{
  \position{Assistant Professor}
  \department{Computer Science}
  \institution{University of California, Davis}
  \city{Davis}
  \state{California}
  \postcode{95616}
  \country{U.S.A.}
}
\email{avthakur@ucdavis.edu}

\begin{abstract}
  Abstract interpretation is a general framework for expressing static program
analyses. It reduces the problem of extracting properties of a program to
computing an approximation of the least fixpoint of a system of equations. The
de facto approach for computing this approximation uses a
sequential algorithm based on weak topological order~(WTO). This paper presents
a deterministic parallel algorithm for fixpoint computation by introducing the
notion of weak partial order~(WPO). We present an algorithm for constructing a
WPO in almost-linear time. Finally, we describe $\pikosname$, our deterministic
parallel abstract interpreter, which extends the sequential abstract interpreter
$\ikos$. We evaluate the performance and scalability of $\pikosname$ on a suite
of $\totalbench$ C programs. When using 4 cores, $\pikosname$ achieves an
average speedup of \ameanspeed{} over $\ikos$, with a maximum speedup of
\maxspeed. When using 16 cores, $\pikosname$ achieves a maximum speedup of \bspeed.

\end{abstract}



\maketitle

\section{Introduction}
\label{sec:introduction}

 Program analysis is a widely adopted approach for automatically extracting
 properties of the dynamic behavior of programs
 \cite{IFM:BCLR04,brat2005precise,SAW:JJA08,SAS:DJ07,FMCAD:Balakrishnan10}.
 Program analyses are used, for instance, for program optimization, bug finding,
 and program verification. To be effective, a program analysis needs to be
 efficient, precise, and deterministic (the analysis always computes the same
 output for the same input program)~\cite{CACM:BBC10}. This paper aims to
 improve the efficiency of program analysis without sacrificing precision or
 determinism.
 
\emph{Abstract interpretation}~\cite{kn:CC77} is a general framework for
 expressing static program analyses. 
 A typical use of abstract interpretation to determine program invariants
 involves:%
 \begin{ecomponents}
 \item
 \label{item:abstractdomain}
  \emph{An abstract domain $\AbsDomain$ that captures relevant program
  properties.} Abstract domains have been developed to perform, for instance,
  numerical analysis \cite{CousotPOPL1978,MineESOP2004,MineHOSC2006,venet2012gauge,oulamara2015abstract,singh2017fast},
  heap analysis \cite{wilhelm2000shape,rinetzky2005semantics}, 
  and information flow \cite{giacobazzi2004abstract}.
 
 \item
 \label{item:equations}
  \emph{An equation system $\Value = F(\Value)$ over
  $\AbsDomain$ that captures the abstract program behavior}:
  \begin{equation}
    \Value_1 = F_1(\Value_1, \ldots, \Value_n), \quad
    \Value_2 = F_1(\Value_1, \ldots, \Value_n), \quad \ldots, \quad
    \Value_n = F_n(\Value_1, \ldots, \Value_n)
   \label{eq:equations}
   \end{equation}
  Each index $i\in [1,n]$ corresponds to a control point of the program,
  the unknowns $\Value_i$ of the system correspond to the invariants to be
  computed for these control points, and each $F_i$ is a monotone operator
  incorporating the abstract transformers and control flow of the program.
 
 \item
 \label{item:fixedpoint}
 \emph{Computing an approximation of the least fixpoint of \pref{eq:equations}.}
  The exact least solution of the system can be computed using Kleene iteration
  starting from the least element of $\AbsDomain^n$ provided $\AbsDomain$ is
  Noetherian. However, most interesting abstract domains require the 
  use of \emph{widening} to ensure termination, which may result in 
  an over-approximation of the invariants of the program. A subsequent
  \emph{narrowing} iteration tries to improve the post solution via a
  downward fixpoint iteration. In practice, abstract interpreters compute an
  \emph{approximation of the least fixpoint}. In this paper, we use
  ``fixpoint'' to refer to such an
  approximation of the least fixpoint.
 \end{ecomponents}
 
 The \emph{iteration strategy} specifies the order in which the equations in
 \pref{eq:equations} are applied during fixpoint computation and where widening
 is performed. For a given abstraction, the efficiency, precision, and determinism
 of an abstract interpreter depends on the iteration strategy. The iteration
 strategy is determined by the dependencies between the individual equations in
 \pref{eq:equations}. If this dependency graph is acyclic, then the optimal
 iteration strategy is any topological order of the vertices in the graph. This
 is not true when the dependency graph contains cycles. Furthermore, each cycle
 in the dependency graph needs to be cut by at least one widening point.

 Since its publication, Bourdoncle's
 algorithm~\cite{bourdoncle1993efficient} has become the de facto approach for
 computing an efficient iteration strategy for abstract interpretation.
 Bourdoncle's algorithm determines the iteration strategy from a \emph{weak
 topological order (WTO)} of the vertices in the dependency graph corresponding to
 the equation system. However, there are certain disadvantages to Bourdoncle's algorithm:
 (i)~the iteration strategy computed by Bourdoncle's algorithm is inherently
 sequential: WTO gives a total order of the vertices in the dependency
 graph; (ii)~computing WTO using Bourdoncle's algorithm has a worst-case
 cubic time complexity; (iii)~the mutually-recursive nature of Bourdoncle's
 algorithm makes it difficult to understand (even for seasoned practitioners of
 abstract interpretation); and (iv)~applying Bourdoncle's algorithm, as is, to
 deep dependency graphs can result in a stack overflow in
 practice.\footnote{\emph{Workaround that prevents stack overflows in WTO computation}
 \url{https://github.com/facebook/redex/commit/6bbf8a5ddb}.}\textsuperscript{,}\footnote{\emph{Possible
 stack overflow while computing WTO of a large CFG}
 \url{https://github.com/seahorn/crab/issues/18}.}

 This paper addresses the above disadvantages of Bourdoncle's algorithm by
 presenting a concurrent iteration strategy for fixpoint computation in an
 abstract interpreter~(\pref{sec:fixpoint}). This concurrent fixpoint
 computation can be efficiently executed on modern multi-core hardware. The
 algorithm for computing our iteration strategy has a worst-case almost-linear
 time complexity, and lends itself to a simple iterative
 implementation~(\pref{sec:Algorithm}). The resulting parallel abstract
 interpreter, however, remains deterministic: for the same program, all possible
 executions of the parallel fixpoint computation give the same result. In fact,
 the fixpoint computed by our parallel algorithm is the same as that computed by
 Bourdoncle's sequential algorithm~(\pref{sec:wto}).

 To determine the concurrent iteration strategy, this paper introduces the
 notion of a \emph{weak partial order (WPO)} for the dependency graph of the
 equation system~(\pref{sec:AxiomaticWPO}). WPO generalizes the notion of WTO: a
 WTO is a linear extension of a WPO~(\pref{sec:wto}). Consequently, the
 almost-linear time algorithm for WPO can also be used to compute WTO. 
 The algorithm for WPO construction handles dependency graphs that are 
 irreducible \cite{hecht1972flow,tarjan1973testing}.
 The key
 insight behind our approach is to adapt algorithms for computing \emph{loop
 nesting forests}~\cite{ramalingamTOPLAS1999,ramalingamTOPLAS2002} to the
 problem of computing a concurrent iteration strategy for abstract
 interpretation.

We have implemented our concurrent fixpoint iteration strategy in a tool called
$\pikosname$~(\pref{sec:implementation}).
Using a suite of $\totalbench$ C
programs, we compare the performance of $\pikosname$ 
against the
state-of-the-art abstract interpreter $\ikos$~\cite{ikos2014}, which
uses Bourdoncle's algorithm (\pref{sec:evaluation}).
When using 4 cores, $\pikosname$ achieves an average speedup of \ameanspeed{}
over $\ikos$, with a maximum speedup of \maxspeed. 
We see that $\pikosname$ exhibits a larger speedup when 
analyzing programs that took longer to analyze using $\ikos$.
$\pikosname$ achieved an average speedup of 1.73x
on programs for which $\ikos$ took less than 16~seconds, 
while $\pikosname$ achieved an average speedup of 2.38x on programs 
for which $\ikos$ took greater than 508~seconds.
The scalability of
$\pikosname$ depends on the structure of the program being analyzed.
When using 16 cores,
$\pikosname$ achieves a maximum speedup of \bspeed.

The contributions of the paper are as follows:
\begin{itemize}
   \item We introduce the notion of a weak partial order (WPO) for a directed
   graph (\pref{sec:AxiomaticWPO}), and show how this generalizes the existing
   notion of weak topological order (WTO)~(\pref{sec:wto}).
   \item We present a concurrent algorithm for computing the fixpoint of a set
   of equations (\pref{sec:fixpoint}). 
   \item We present an almost-linear time algorithm for WPO and
   WTO construction
   (\pref{sec:Algorithm}).
   \item We describe our deterministic parallel abstract interpreter $\pikosname$ (\pref{sec:implementation}),
   and evaluate its performance on a suite of C programs (\pref{sec:evaluation}).
\end{itemize}
\pref{sec:overview} presents an overview of the technique;
\pref{sec:preliminaries} presents mathematical preliminaries;
\pref{sec:related} describes related work; \pref{sec:conclusion}
concludes.

\section{Overview}
\label{sec:overview}

Abstract interpretation is a general framework that captures most existing
approaches for static program analyses and reduces extracting properties of
programs to approximating their semantics \cite{kn:CC77,CousotGR:POPL2019}.
Consequently, this section is not meant to capture all possible approaches to
implementing abstract interpretation or describe all the complex optimizations
involved in a modern implementation  of an abstract interpreter. Instead it is
only meant to set the appropriate context for the rest of the paper, and to
capture the relevant high-level structure of abstract-interpretation
implementations such as $\ikos$ \cite{ikos2014}.

\subsubsubsection{Fixpoint equations.}
Consider the simple program \texttt{P} represented by its control flow graph
(CFG) in \pref{fig:ExampleProgram}. 
We will illustrate how an abstract
interpreter would compute the set of values that variable \texttt{x} 
might contain at each program point \texttt{i} in \texttt{P}.
In this example, we will use the standard integer interval domain \cite{ISP:CC76,kn:CC77}
represented by the complete lattice $\langle Int, \sqsubseteq, \bot,\top, \sqcup, \sqcap \rangle$
with $Int \eqdef \{\bot\}
                  \cup \{ [l, u] \mid l,u \in \IntegerSet \land l \leq u \} 
                  \cup \{ [-\infty, u] \mid u \in \IntegerSet \} 
                  \cup \{ [l, \infty] \mid l \in \IntegerSet \} 
                  \cup \{[-\infty, \infty]\} $. 
The partial order  $\sqsubseteq$ on $Int$ is interval inclusion with
the empty interval $\bot = \emptyset$ encoded as $[\infty, -\infty]$ and 
$\top = [-\infty, \infty]$.

\pref{fig:ExampleEquationSystem} shows the corresponding equation system $\Value
= F(\Value)$, where $\Value = (\Value_0, \Value_1, \ldots, \Value_8)$. Each
equation in this equation system is of the form $\Value_i = F_i(\Value_0,
\Value_1, \ldots, \Value_8)$, where the variable $\Value_i \in Int$ represents
the interval value at program point $i$ in \texttt{P} and $F_i$ is monotone. The
operator~$+$ represents the (standard) addition operator over $Int$. As is
common (but not necessary), the dependencies among the equations reflect the CFG
of program \texttt{P}.

The exact least solution $\Value^{\textrm{lfp}}$ of the equation system $\Value
= F(\Value)$ would give the required set of values for variable \texttt{x} at
program point $i$. Let $\Value^0 = (\bot, \bot, \ldots, \bot)$ and $\Value^{i+1}
= F(\Value^i), i \geq 0$ represent the standard Kleene iterates, which converge
to $\Value^{\textrm{lfp}}$. 

\subsubsubsection{Chaotic iteration.}
Instead of applying the function $F$ during Kleene
iteration, one can use \emph{chaotic iterations}
\cite{cousot1977asynchronous,kn:CC77} and apply the individual equations $F_i$.
The order in which the individual equations are applied is determined by the
\emph{chaotic iteration strategy}.

\subsubsubsection{Widening.}
For non-Noetherian abstract domains, such as the interval abstract domain,
termination of this Kleene iteration sequence requires the use of a
\emph{widening operator}~($\widen$)~\cite{kn:CC77,Cousot:VMCAI2015}. A set of
\emph{widening points} $W$ is chosen and the equation for $i \in W$ is replaced
by $\Value_i = \Value_i \widen F_i(\Value_0, \ldots, \Value_n)$. An
\emph{admissible set of widening points} ``cuts'' each cycle in the dependency
graph of the equation system by the use of a widening operator to ensure
termination~\cite{kn:CC77}. Finding a minimal admissible set of widening points
is an NP-complete problem~\cite{garey2002computers}. A~possible widening
operator for the interval abstract domain is defined by:
$\bot \widen I = I \widen \bot = I \in Int$ and $[i, j] \widen [k, l] =
[\textrm{if } k <i \textrm{ then } -\infty
\textrm{ else } i, \textrm{if } l > j \textrm{ then } \infty \textrm{ else }
j]$. This widening operator is non-monotone. The application of a widening
operator may result in a crude over-approximation of the least fixpoint; more
sophisticated widening operators as well as techniques such as narrowing can be
used to ensure
precision~\cite{kn:CC77,DBLP:conf/cav/GopanR06,amatoSAS2013,amatoSCP2016,DBLP:journals/spe/KimHOY16}. %
Although the discussion of our fixpoint algorithm uses a simple widening
strategy~(\pref{sec:fixpoint}), our implementation incorporates more
sophisticated widening and narrowing strategies implemented in
$\ikos$~(\pref{sec:implementation}).

\begin{figure}
  \centering
  \begin{subfigure}[b]{.4\textwidth} 
    \centering
    \begin{tikzpicture}[auto,node distance=.9cm,]
      \tikzstyle{every node} = [rectangle, draw, inner sep=1pt,minimum size=2.5ex]
      \node [minimum width=6ex] (entry) at (0,0) {$0\colon\quad$};
      \node [below of=entry] (1) {$1\colon\texttt{x=0}$};
      \node [below of=1, minimum width=6ex] (2) {$2\colon\quad$};
      \node [below of=2] (3) {$3\colon \texttt{x=x+1}$};
      \node [right of=1, node distance=1.8cm] (8) {$8\colon\texttt{x=1}$};
      \node [right of=2, node distance=1.8cm, minimum width=6ex] (4) {$4\colon\quad$};
      \node [below of=4] (5) {$5\colon \texttt{x=x+1}$};
      \node [right of=5, node distance=1.5cm] (6) {$6\colon\texttt{x=0}$};
      \node [right of=4, node distance=1.5cm, minimum width=6ex] (7) {$7\colon\quad$};

      \path (entry) edge[cfgedge] (1);
      \path (entry) edge[cfgedge] (8);
      \path (1) edge[cfgedge] (2);
      \path (2) edge[cfgedge] (3);
      \path (2) edge[cfgedge] (4);
      \path (8) edge[cfgedge] (4);
      \path (4) edge[cfgedge] (5);
      \path (4) edge[cfgedge] (6);
      \path (4) edge[cfgedge] (7);
      \path (3) edge[cfgedge, bend left=45] (2);
      \path (5) edge[cfgedge, bend left=45] (4);
      \path (6) edge[cfgedge, bend right=25] (4);%
    \end{tikzpicture}%
    \caption{}
    \label{fig:ExampleProgram}
  \end{subfigure}
  \begin{subfigure}[b]{.28\textwidth} 
  {\small %
  \begingroup %
  \zerodisplayskips
  \begin{align*}
    \Value_0 &= \top \\
    \Value_1 &= [0,0] \\
    \Value_2 &= \Value_1 \sqcup \Value_3 \\
    \Value_3 &= \Value_2 + [1,1] \\
    \Value_4 &= \Value_2 \sqcup \Value_8 \sqcup 
              \Value_5 \sqcup \Value_6 \\
    \Value_5 &= \Value_4 + [1,1]\\
    \Value_6 &= [0,0] \\
    \Value_7 &= \Value_4 \\
    \Value_8 &= [1,1] %
  \end{align*}%
  \endgroup
  }%
  \caption{}
  \label{fig:ExampleEquationSystem}
  \end{subfigure}
  \begin{subfigure}[b]{.3\textwidth} 
    {\small %
    \begingroup %
    \zerodisplayskips
    \begin{align*}
      \Value_0 &= \top \\
      \Value_1 &= [0,0] \\
      \Value_2 &= \Value_2 \widen (\Value_1 \sqcup \Value_3) \\
      \Value_3 &= \Value_2 + [1,1] \\
      \Value_4 &= \Value_4 \widen (\Value_2 \sqcup \Value_8 \sqcup 
                \Value_5 \sqcup \Value_6) \\
      \Value_5 &= \Value_4 + [1,1]\\
      \Value_6 &= [0,0] \\
      \Value_7 &= \Value_4 \\
      \Value_8 &= [1,1] %
    \end{align*}%
    \endgroup
    }%
    \caption{}
    \label{fig:ExampleEquationSystemWithWidening}
  \end{subfigure}
  \vspace{-2ex}
  \caption{\subref{fig:ExampleProgram}~A simple program \texttt{P} 
  that updates \texttt{x}; \subref{fig:ExampleEquationSystem}~Corresponding equation system
  for interval domain; 
  \subref{fig:ExampleEquationSystemWithWidening}~Corresponding equation system with 
  vertices $2$ and $4$ as widening points.}
\end{figure}

\subsubsubsection{Bourdoncle's approach.} 
\citet{bourdoncle1993efficient} 
introduces the notion of \emph{hierarchical total order (HTO)} of a set and
\emph{weak topological order (WTO)} of a directed graph~(see \pref{sec:wto}). An
admissible set of widening points as well as a chaotic iteration strategy,
called the \emph{recursive strategy}, can be computed using a WTO of the
dependency graph of the equation system.
A WTO for the equation system in
\pref{fig:ExampleEquationSystem} is $T \eqdef 0\ 8\ 1\ (2\ 3)\ (4\ 5\ 6)\ 7$. The set of
elements between two matching parentheses are called a \emph{component of the WTO},
and the first element of a component is called the \emph{head of the component}.
Notice that components are non-trivial strongly connected components (``loops'')
in the directed graph of \pref{fig:ExampleProgram}. 
\citet{bourdoncle1993efficient} proves that the set of component heads is an
admissible set of widening points. For \pref{fig:ExampleEquationSystem}, the set
of heads $\{2, 4\}$ is an admissible set of widening points.
\pref{fig:ExampleEquationSystemWithWidening} shows the corresponding equation
system that uses widening. 

The iteration strategy generated using WTO $T$ is $S_1 \eqdef 0\ 8\ 1\ [2\ 3]^*\
[4\ 5\ 6]^*\ 7$, where occurrence of~$i$ in the sequence represents applying the
equation for $\Value_i$, and $[\ldots]^*$ is the ``iterate until stabilization''
operator. A component is \emph{stabilized} if iterating over its elements 
does not change their values. The component heads $2$ and $4$ are chosen as
 widening points. The iteration sequence $S_1$ should be interpreted as
``apply equation for $\Value_0$, then apply the equation for $\Value_8$, then
apply the equation for $\Value_1$, repeatedly apply equations for $\Value_2$ and
$\Value_3$ until stabilization'' and so on. Furthermore,
\citet{bourdoncle1993efficient} showed that stabilization of a component can be
detected by the stabilization of its head. For instance, stabilization of
component $\{2, 3\}$ can be detected by the stabilization of its head $2$. This
property minimizes the number of (potentially expensive) comparisons between
abstract values during fixpoint computation. For the equation system of
\pref{fig:ExampleEquationSystemWithWidening}, the use of Bourdoncle's recursive
iteration strategy would give us $\Value_7^{\textrm{fp}} = [0, \infty]$.

\subsubsubsection{Asynchronous iterations.}
The iteration strategy produced by Bourdoncle's approach is necessarily
sequential, because the iteration sequence is generated from a total order. One
could alternatively implement a parallel fixpoint computation using
\emph{asynchronous iterations} \cite{cousot1977asynchronous}: each processor
$i$ computes the new value of $\Value_i$  accessing the shared state consisting
of $\Value$ using the appropriate equation from
\pref{fig:ExampleEquationSystemWithWidening}. However, the parallel fixpoint
computation using asynchronous iterations is non-deterministic; that is, the
fixpoint computed might differ based on the order in which the equations are applied
(as noted by \citet{monniaux2005parallel}). The reason for this non-determinism is due to the
non-monotonicity of widening. For example, if the iteration sequence
$S_2 \eqdef 0\ 8\ [4\ 5\ 6]^*\ 1\ [ 2\ 3 ]^*\ [4\ 5\ 6]^*\ 7$ is used to compute
the fixpoint for the equations in \pref{fig:ExampleEquationSystemWithWidening},
then $\Value_7^{\textrm{fp}} = [-\infty, \infty]$, which differs from the
value computed using iteration sequence~$S_1$. 

\subsubsubsection{Our deterministic parallel fixpoint computation.}
In this paper, we present a parallel fixpoint computation that is
\emph{deterministic}, and, in fact, gives the \emph{same result} as Bourdoncle's
sequential fixpoint computation~(\pref{sec:wto}). Our approach generalizes
Bourdoncle's hierarchical total order and weak topological order to
\emph{hierarchical partial order~(HPO)} and \emph{weak partial
order~(WPO)}~(\pref{sec:wpo}). The iteration strategy is then based on the WPO of
the dependency graph of the equation system. The use of partial orders, instead
of total orders, enables us to generate an iteration strategy that is
concurrent~(\pref{sec:fixpoint}). For the equation system in
\pref{fig:ExampleEquationSystemWithWidening}, our approach would produce the
iteration sequence represented as $S_3 \eqdef  0\ ((1\ [2\ 3]^*) \mid 8)\ [4\
(5\mid 6)]^*\ 7$, where $\mid$ represents concurrent execution. Thus, the
iteration (sub)sequences $1\ [2\ 3]^*$ and $8$ can be computed in parallel, as
well as the subsequences $5$ and $6$. However, unlike iteration sequence $S_2$,
the value for $\Value_4$ in $S_3$ will not be computed until the component $\{
2, 3\}$ stabilizes. Intuitively, determinism is achieved by ensuring that no
element outside the component will read the value of an element in the component
until the component stabilizes. In our algorithm, the value of $2$ is read by
elements outside of the component $\{2, 3\}$ only after the component
stabilizes. Similarly, the value of $4$ will be read by $7$ only after the
component $\{4, 5, 6\}$ stabilizes. Parallel fixpoint computation based on a WPO
results in the same fixpoint as the sequential computation based on a
WTO~(\pref{sec:wto}).

\section{Mathematical Preliminaries}
\label{sec:preliminaries}

A \emph{binary relation} $\relation$ on set $S$ is a subset of the Cartesian product of
$S$ and $S$; that is, $\relation \subseteq S \times S$. Given $S' \subseteq S$,
let $\relation_{\downharpoonright S'} = \relation \cap (S' \times S')$.
A relation
$\relation$ on set $S$ is said to be \emph{one-to-one} iff for all $w,x,y, z \in
S$, $(x, z) \in \relation$ and $(y,z) \in \relation$ implies $x = y$, and $(w,x)
\in \relation$ and $(w, y) \in \relation$ implies $x = y$.
A \emph{transitive closure} of a binary relation $\relation$, denoted by
$\relation^+$, is the smallest transitive binary relation that
contains $\relation$.
A \emph{reflexive transitive closure} of a binary relation $\relation$, denoted
by $\relation^*$, is the smallest reflexive transitive binary
relation that contains $\relation$.

A \emph{preorder} $(S, \relation)$ is a set $S$ and a binary relation
$\relation$ over $S$ that is reflexive and transitive. A \emph{partial order}
$(S, \relation)$ is a preorder where $\relation$ is antisymmetric. Two elements
$u, v \in S$ are \emph{comparable} in a partial order $(S, \relation)$ if
$(u,v)\in \relation$ or $(v, u) \in \relation$. A \emph{linear (total) order} or
\emph{chain} is a partial order in which every pair of its elements are
comparable.
A partial order $(S, \relation')$ is an \emph{extension} of a partial order $(S,
\relation)$ if $\relation \subseteq \relation'$; an extension that is a linear
order is called a \emph{linear extension}. There exists a linear extension for
every partial order \cite{szpilrajn1930extension}.

Given a partial order $(S, \relation)$, define $\postset{x}_{\relation} \eqdef
\{ y\in S \mid (x, y) \in R \}$, and $\preset{x}_{\relation} \eqdef \{ v \in S
\mid (v, x) \in R \}$, and $\component{x}{y}_{\relation} \eqdef
\postset{x}_{\relation} \cap \preset{y}_{\relation}$.
A partial order $(S, \relation)$ is a \emph{forest} if for all $x \in S$,
$(\preset{x}_{\relation}, R)$ is a chain.

\begin{example}
  \label{exa:ExamplePartialOrder}
  Let $(Y, \mathsf{T})$ be a partial order with 
  $Y = \{y_1, y_2, y_3, y_4 \}$ 
  and $\mathsf{T} = \{ (y_1, y_2), (y_2, y_3),$ $(y_2, y_4) \}^*$. 
  Let $Y' = \{ y_1, y_2 \} \subseteq Y$, then $\mathsf{T}_{\downharpoonright Y'} = \{(y_1, y_1), (y_1, y_2), (y_2, y_2) \}$.
  \begin{align*}
    \postset{y_1}_{\mathsf{T}} &= \{y_1, y_2, y_3, y_4\}&  \postset{y_2}_{\mathsf{T}} &= \{y_2, y_3, y_4\} & \postset{y_3}_{\mathsf{T}} &=  \{y_3\}       & \postset{y_4}_{\mathsf{T}} &= \{ y_4\} \\
    \preset{y_1}_{\mathsf{T}}  &= \{y_1 \}             &  \preset{y_2}_{\mathsf{T}} &= \{y_1, y_2\}       & \preset{y_3}_{\mathsf{T}} &= \{y_1, y_2, y_3\} & \preset{y_4}_{\mathsf{T}} &= \{y_1, y_2, y_4\} 
  \end{align*}
  We see that the partial order $(Y, \mathsf{T})$ is a forest because for all $y \in Y$, $(\preset{y}_{\mathsf{T}}, \mathsf{T})$ is a chain.
  \begin{align*}
    \component{y_1}{y_1}_{\mathsf{T}} &= \{y_1\}       &  \component{y_1}{y_2}_{\mathsf{T}} &= \{y_1, y_2\} & \component{y_1}{y_3}_{\mathsf{T}} &= \{y_1, y_2, y_3\} & \component{y_1}{y_4}_{\mathsf{T}} &= \{y_1, y_2, y_4\} \\
    \component{y_4}{y_1}_{\mathsf{T}} &= \emptyset &   \component{y_4}{y_2}_{\mathsf{T}} &= \emptyset    &  \component{y_4}{y_3}_{\mathsf{T}} &= \emptyset  & \component{y_4}{y_4}_{\mathsf{T}} &= \{y_4\}  
  \end{align*} 
  \qef   
  \vspace{-1ex}
\end{example}
A \emph{directed graph} $G(V, \cfgarrow)$ is defined by a set of vertices $V$
and a binary relation $\cfgarrow$ over $V$. 
The reachability among vertices is captured by the preorder~$\cfgarrow^*$:
there is a path from
vertex $u$ to vertex $v$ in $G$ iff $u\cfgarrow^* v$. $G$ is a \emph{directed acyclic graph (DAG)}
iff $(V, \cfgarrow^*)$ is a partial order. A topological order of a DAG $G$
corresponds to a linear extension of the partial order $(V, \cfgarrow^*)$. We use
$G_{\downharpoonright V'}$ to denote the subgraph $(V \cap V',
\cfgarrow_{\downharpoonright V'})$.
Given a directed graph $G(V, \cfgarrow)$, a \emph{depth-first numbering (DFN)}
is the order in which vertices are discovered during a depth-first search (DFS)
of $G$. A \emph{post depth-first numbering (post-DFN)} is the order in
which vertices are finished during a DFS of $G$. A \emph{depth-first tree (DFT)}
of $G$ is a tree formed by the edges used to discover vertices
during a DFS. Given a DFT of $G$, an edge $u \cfgarrow v$ is called (i)~a \emph{tree edge} 
if $v$ is a child of $u$ in the DFT; (ii)~a \emph{back edge} if $v$ is an ancestor of $u$ 
in the DFT; (iii)~a \emph{forward edge} if it is not a tree edge and $v$ is a descendant of $u$ 
in the DFT; and (iv)~a \emph{cross edge} otherwise \cite{CLRS}.
In general, a directed graph might contain multiple connected components and 
a DFS yields a \emph{depth-first forest} (DFF).
The \emph{lowest common ancestor (LCA)} of vertices $u$ and $v$ in a rooted tree $T$ is a vertex 
that is an ancestor of both $u$ and $v$ and that has the greatest depth in $T$ \cite{tarjan1979olca}.
It is unique for all pairs of vertices.

A \emph{strongly connected component (SCC)} of a directed graph $G(V, \cfgarrow)$
is a subgraph of $G$ such that $u\cfgarrow^* v$ for all $u, v$ in the subgraph.
An SCC is trivial if it only consists of a single vertex without any edges.
A \emph{feedback edge set} $B$ of a graph $G(V, \cfgarrow)$ is a subset of $\cfgarrow$ such
that $(V, (\cfgarrow \setminus B)^*)$ is a partial order; that is, 
the directed graph $G(V, \cfgarrow \setminus B)$ is a DAG. The problem
of finding the minimum feedback edge set is NP-complete \cite{Karp1972}.

\begin{example}
  Let $G(V, \cfgarrow)$ be directed graph shown in \pref{fig:ExampleProgram}.
  The ids used to label the vertices $V$ of $G$ correspond to a depth-first
  numbering (DFN) of the directed graph $G$. The following lists the vertices in
  increasing post-DFN numbering: $3, 5, 6, 7, 4 ,2,1, 8, 0$. Edges $(3,2)$,
  $(5,4)$, and $(6,4)$ are back edges for the DFF that is assumed by the DFN,
  edge $(8,4)$ is a cross edge, and the rest are tree edges. The lowest common
  ancestor (LCA) of $3$ and $7$ in this DFF is $2$. The subgraphs induced by the
  vertex sets $\{2, 3\}$, $\{4, 5\}$, $\{4, 6\}$, and $\{4, 5, 6\}$ are all
  non-trivial SCCs. The minimum feedback edge set of $G$ is $F = \{(3,2), (5,4),
  (6,4)\}$. We see that the graph $G(V, \cfgarrow \setminus F)$ is a DAG.
  \qef
\end{example}

\section{Axiomatic Characterization of Weak Partial Order}
\label{sec:AxiomaticWPO}
\label{sec:wpo}

This section introduces the notion of Weak Partial Order (WPO), presents its
axiomatic characterization, and proves relevant properties. A constructive
characterization is deferred to~\pref{sec:Algorithm}. 
The notion of WPO is built
upon the notion of a hierarchical partial order, which we define first.

A \emph{Hierarchical Partial Order (HPO)} is a partial order $(S, \preceq)$
overlaid with a nesting relation $\nesting \subseteq S \times S$ that structures
the elements of $S$ into well-nested hierarchical \emph{components}. As
we will see in \pref{sec:fixpoint}, the elements in a component are
iterated over until stabilization in a fixpoint iteration strategy, and the
partial order enables concurrent execution.

\newpage
\begin{definition}
  \label{def:HPO}
  A \emph{hierarchical partial order} $\hpo$ is a 3-tuple $(S, \preceq,
  \nesting)$ such that:
  \begin{enumerate}[label={H\arabic*.}, ref={H\arabic*}]
    \item \label{it:Hpartial}
      $(S, \preceq)$ is a partial order.

    \item \label{it:Hoto}
      $\nesting \subseteq S \times S$ is one-to-one.

    \item \label{it:Hdir}
      $(x, h) \in \nesting$ implies $h \prec x$.
    
    \item \label{it:Hnest} 
      Partial order $(\components_{\hpo}, \subseteq)$  is a forest, where
      $\components_\hpo \eqdef \big\{\component{h}{x}_{\preceq} \mid (x,h) \in
      \nesting \big\}$ is the set of components.

    \item \label{it:Hexit}
      For all $h, x, u, v \in S$, 
      $ h \preceq u \preceq x$ and $(x, h) \in \nesting$ and $u \preceq v$ 
      implies either $x \prec v$ or $v \preceq x$.
    \qef
  \end{enumerate} 
\end{definition}

For each $(x,h) \in \nesting$, the set $\component{h}{x}_{\preceq}
\eqdef \{u \in S \mid h \preceq u \preceq x\}$ defines a \emph{component} of the
HPO, with $x$ and $h$ referred to as the \emph{exit} and \emph{head} of the
component. A component can be identified using either its head or its exit due
to condition \ref{it:Hoto}; we use $C_h$ or $C_x$ to denote a component
with head $h$ and exit $x$.
 Condition \ref{it:Hdir} states that the nesting relation $\nesting$ is in the opposite 
 direction of the partial order $\preceq$. The reason for this convention will be clearer when 
 we introduce the notion of WPO (\pref{def:WPO}), where we show that the nesting relation $\nesting$ 
 has a connection to the feedback edge set of the directed graph. 
Condition \ref{it:Hnest} implies that the set of components $\components_\hpo$ is
well-nested; that is, two components should be either mutually disjoint or
one must be a subset of the other. 

Condition \ref{it:Hexit} states that if an element $v$ depends upon an element
$u$ in a component $C_x$, then either $v$ depends on the exit $x$ or $v$ is in
the component $C_x$. Recall that $(x, h) \in \nesting$ and $h \preceq u \preceq
x$ implies $u \in C_x$ by definition. Furthermore, $v \preceq x$ and $u\preceq v$
implies $v \in C_x$.  Condition \ref{it:Hexit} ensures determinism of the
concurrent iteration strategy (\pref{sec:fixpoint}); this
condition ensures that the value of $u$ does not ``leak'' from $C_x$ during
fixpoint computation until the component $C_x$ stabilizes.

\begin{example}
  Consider the partial order $(Y, \mathsf{T})$ defined in
  \pref{exa:ExamplePartialOrder}. 
  Let $\nesting_1 = \{(y_3, y_1), (y_4,y_2) \}$. $(Y, \mathsf{T}, \nesting_1)$
  violates condition \ref{it:Hnest}. In particular, the components $C_{y_3} =
  C_{y_1} =  \{y_1, y_2, y_3\}$ and $C_{y_4} = C_{y_2} = \{ y_2, y_4\}$ are
  neither disjoint nor is one a subset of the other. Thus, $(Y, \mathsf{T},
  \nesting_1)$ is \emph{not} an HPO.
  
  Let $\nesting_2 = \{(y_3, y_1)\}$. $(Y, \mathsf{T}, \nesting_2)$ violates
  condition \ref{it:Hexit}. In particular, $y_2 \in C_{y_3}$ and $(y_2, y_4)\in
  \mathsf{T}$, but we do not have $y_3 \prec y_4$ or $y_4 \preceq y_3$.  Thus,
  $(Y, \mathsf{T}, \nesting_2)$ is \emph{not} an HPO.

  Let $\nesting_3 = \{(y_2, y_1)\}$. $(Y, \mathsf{T}, \nesting_3)$ is an HPO satisfying 
  all conditions \ref{it:Hpartial}--\ref{it:Hexit}.
  \qef
\end{example}

Building upon the notion of an HPO, we now define a \emph{Weak Partial Order
(WPO)} for a directed graph $G(V, \cfgarrow)$. In the context of fixpoint
computation, $G$ represents the dependency graph of the fixpoint equation
system. To find an effective iteration strategy, the cyclic dependencies in $G$
need to be broken. In effect, a WPO partitions the preorder $\cfgarrowstar$ into
a partial order $\forwardarrowstar$ and an edge set defined using of a nesting
relation $\backwardarrow$.  

\begin{definition}
  \label{def:WPO}
  A \emph{weak partial order} $\wpo$ for a directed graph $G(V, \cfgarrow)$
  is a 4-tuple $(V, X, \forwardarrow, \backwardarrow)$ 
  such that:
  \begin{enumerate}[label={W\arabic*.}, ref={W\arabic*}]
    \item \label{it:Wx} \label{it:VcapX}
      $V \cap X = \emptyset$.
    \item \label{it:Wred}
      $\backwardarrow \subseteq X \times V$, and for all $x \in X$, there exists $v
      \in V$ such that $x \backwardarrow v$.
    \item \label{it:Wblue} \label{it:forward}
      $\forwardarrow \subseteq (V \cup X) \times (V \cup X)$.
    \item \label{it:Whpo} 
      $\hpo (V \cup X, \forwardarrowstar, \backwardarrow)$ is a hierarchical partial
      order (HPO).
    \item \label{it:Wbd} 
      For all $u \cfgarrow v$, either 
      (i)~$u \forwardarrowplus v$, or
      (ii)~$u \in \component{v}{x}_{\subforwardarrowstar}$ and
      $x\backwardarrow v$ for some $x \in X$. \qef
  \end{enumerate}
\end{definition}

Condition \ref{it:Whpo} states that $\hpo (V \cup X, \forwardarrowstar,
\backwardarrow)$  is an HPO. Consequently, $(V \cup X, \forwardarrowstar)$ is a
partial order and $\backwardarrow$ plays the role of the nesting relation
in~\pref{def:HPO}. We refer to the relation $\forwardarrow$ in WPO $\wpo$ as the
\emph{\forward{}s}, the relation $\backwardarrow$ as \emph{\backward{}s}, and
the set $X$ as the \emph{exits}. Furthermore, the notion of components
$\components_\hpo$ of an HPO $\hpo$ as defined in \pref{def:HPO} can be lifted
to \emph{components of WPO} $\components_\wpo \eqdef
\big\{\component{h}{x}_{\subforwardarrowstar} \mid x \backwardarrow h \big\}$.
Condition \ref{it:Whpo} ensures that the concurrent iteration strategy for the
WPO is deterministic (\pref{sec:fixpoint}).
Condition~\ref{it:Wx} states that exits are always new.
Condition~\ref{it:Wred} states that $X$ does not contain any unnecessary elements. 

Condition \ref{it:Wbd} connects the relation $\cfgarrow$ of the directed graph
$G$ with relations $\forwardarrow$ and $\backwardarrow$ used in the HPO $\hpo(V \cup X,
\forwardarrowstar, \backwardarrow)$
in condition \ref{it:Whpo}. 
Condition~\ref{it:Wbd} ensures that all dependencies
$u \cfgarrow v$ in $G$ are captured by the HPO $\hpo$ either via a relation in
the partial order $u \forwardarrowplus v$ or indirectly via the component
corresponding to $x\backwardarrow v$, as formalized by the following 
theorem:

\begin{theorem}
  \label{thm:WPODependencies}
  For graph $G(V, \cfgarrow)$ and its WPO $\wpo(V, X,
  \forwardarrow, \backwardarrow)$,
   $\cfgarrowstar \subseteq (\forwardarrow \cup
  \backwardarrow)^*$.
\end{theorem}
\begin{proof}
  By property \ref{it:Wbd}, for each $u \cfgarrow v$, either $u
  \forwardarrowplus v$, or $u \in \component{v}{x}_{\forwardarrowstar}$ and $x
  \backwardarrow v$ for some $x \in X$. For the latter, $u \in
  \component{v}{x}_{\forwardarrowstar}$ implies that $u \forwardarrowplus x
  \backwardarrow v$. Thus, $\cfgarrowstar \subseteq (\forwardarrow \cup
  \backwardarrow)^*$.
\end{proof}

\begin{figure}[t]
  \begin{subfigure}[b]{.33\textwidth}
    \centering
    \begin{tikzpicture}[auto,node distance=.8cm,font=\small]
      \tikzstyle{every node} = [rectangle, draw, inner sep=0pt,minimum size=2.5ex]
      \node (1) {$1$};
      \node [right of=1] (2)  {$2$};
      \node [below right of=2] (6) {$6$};
      \node [right of=6] (7) {$7$};
      \node [below of=7, node distance=0.5cm] (9) {$9$};
      \node [right of=7] (8)  {$8$};
      \node [above right of=2] (3)  {$3$};
      \node [right of=3] (4) {$4$};
      \node [right of=2, node distance=3cm] (5) {$5$};
      \node [below of=2] (10)  {$10$};
      \path (1) edge[cfgedge] (2); 
      \path (2) edge[cfgedge] (3); 
      \path (3) edge[cfgedge] (4); 
      \path (3) edge[cfgedge, bend right=12] ([yshift=-1mm] 5.north west); 
      \path (6) edge[cfgedge, bend left=12] ([yshift=1mm] 5.south west); 
      \path (4) edge[cfgedge, bend right=30] (3); 
      \path (5) edge[cfgedge, bend right=95] (2); 
      \path (2) edge[cfgedge] (6); 
      \path (6) edge[cfgedge] (7); 
      \path (6) edge[cfgedge] (9); 
      \path (7) edge[cfgedge] (8); 
      \path (9) edge[cfgedge] (8); 
      \path (2) edge[cfgedge] (10); 
      \path (8) edge[cfgedge, bend left=90, looseness=1.3] (6); 
    \end{tikzpicture}
    \caption{}
    \label{fig:cfg1}
  \end{subfigure} %
  \begin{subfigure}[b]{.66\textwidth}
    \centering
    \begin{tikzpicture}[auto,node distance=.8cm,font=\small]
      \tikzstyle{every node} = [circle, draw, inner sep=0pt,minimum size=2.5ex]
      \node (1) at (0, 0) {$1$};
      \node [right of=1] (2) {$2$};
      \node [right of=2] (3) {$3$};
      \node [right of=3] (4) {$4$};
      \node [right of=4] (3') {$x_3$};
      \node (5) at (4.7, 0) {$5$};
      \node [right of=5] (2') {$x_2$};
      \node [below of=3, node distance=0.5cm] (6) {$6$};
      \node [right of=6] (7) {$7$};
      \node [right of=7] (8) {$8$};
      \node [below of=7, node distance=0.45cm] (9) {$9$};
      \node [right of=8] (6') {$x_6$};
      \node [right of=2'] (10) {$10$};
      
      \path (1) edge[forward] (2); 
      \path (2) edge[forward] (3); 
      \path (2) edge[forward] (6); 
      \path (3) edge[forward] (4); 
      \path (4) edge[forward] (3'); 
      \path (3') edge[forward] (5); 
      \path (5) edge[forward] (2'); 
      \path (6) edge[forward] (7); 
      \path (6) edge[forward] (9); 
      \path (7) edge[forward] (8); 
      \path (9) edge[forward] (8); 
      \path (8) edge[forward] (6'); 
      \path (6') edge[forward] (5); 
      \path (2') edge[forward] (10); 
      \path (2') edge[backward, bend right=20] (2); 
      \path (3') edge[backward, bend right=30] (3); 
      \path (6') edge[backward, bend right=20] (6); 
    \end{tikzpicture}
    \vspace{0.4cm}
    \caption{} 
    \label{fig:wpo1}
  \end{subfigure}
  \vspace{-5ex}
  \caption{\subref{fig:cfg1}~Directed graph $G_1$ and \subref{fig:wpo1}~WPO 
  $\wpo_1$. Vertices $V$
  are labeled using DFN;
   exits $X = \{x_2, x_3, x_6 \}$. 
   }
  \label{fig:cfg-wpo1}
\end{figure}

\begin{example}
  Consider the directed graph $G_1(V, \cfgarrow)$ in \pref{fig:cfg1}.
  \pref{fig:wpo1} shows a WPO $\wpo_1 (V, X, \forwardarrow, \backwardarrow)$ for
  $G_1$, where $X$ = $\{x_2, x_3, x_6\}$, and satisfies all
  conditions in \pref{def:WPO}. One can verify that $(V \cup X,
  \forwardarrowstar, \backwardarrow)$ satisfies all conditions   in
  \pref{def:HPO} and is an HPO. 

  Suppose we were to remove $x_{6} \forwardarrow 5$ and 
  instead add $6 \forwardarrow 5$ to $\wpo_1$ (to more closely match 
  the edges in $G_1$), then this change would violate condition \pref{it:Hexit}, 
  and hence condition \pref{it:Whpo}.

  If we were to only remove $x_{6} \forwardarrow 5$ from $\wpo_1$, 
  then it would still satisfy condition \pref{it:Whpo}. 
  However, this change would violate condition \pref{it:Wbd}. 
  \qef
\end{example}

\begin{definition}
\label{def:WPOBackedges} 
  For graph $G(V, \cfgarrow)$ and its WPO $\wpo(V, X, \forwardarrow,
  \backwardarrow)$, the \emph{back edges of $G$ with respect to the WPO $\wpo$},
  denote by $B_\wpo$, are defined as $B_\wpo \eqdef \{ (u, v) \in \cfgarrow \mid
  \exists x \in X. 
  u \in \component{v}{x}_{\forwardarrowstar} \land  x \backwardarrow v\}$. \qef
\end{definition}

In other words, $(u,v) \in B_\wpo$ if $u\cfgarrow v$ satisfies
condition~\ref{it:Wbd}-(ii) in \pref{def:WPO}. \pref{thm:feedback} proves that
$B_\wpo$ is a feedback edge set for $G$, and \pref{thm:partial} shows that the
subgraph $(V, \cfgarrow \setminus B_\wpo)$ forms a DAG. Together these two
theorems capture the fact that the WPO $\wpo(V, X, \forwardarrow,
\backwardarrow)$ partitions the preorder~$\cfgarrowstar$ of $G(V, \cfgarrow)$
into a partial order $\forwardarrowstar$ and a feedback edge set $B_\wpo$.

\begin{theorem}
  \label{thm:feedback}
  For graph $G(V, \cfgarrow)$ and its WPO $\wpo(V, X,
  \forwardarrow, \backwardarrow)$,
  $B_\wpo$ is a feedback edge set for $G$.
\end{theorem}
\begin{proof}
  Let $v_1 \cfgarrow v_2 \cfgarrow \dots \cfgarrow v_n \cfgarrow v_1$ be a cycle
  of $n$ distinct vertices in $G$. We will show that there exists $i \in [1,n)$
  such that $v_i \cfgarrow v_{i+1} \in B_\wpo$; that is, $v_i \in
  \component{v_{i+1}}{x}_{\forwardarrowstar}$ and $x \backwardarrow v_{i+1}$ for
  some $x \in X$. If this were not true, then $v_1 \forwardarrowplus \dots
  \forwardarrowplus v_n \forwardarrowplus v_1$ (using~\ref{it:Wbd}). Therefore,
  $v_1 \forwardarrowplus v_n$ and $v_n \forwardarrowplus v_1$, which contradicts
  the fact that $\forwardarrowstar$ is a partial order. Thus, $B_\wpo$ cuts all
  cycles at least once and is a feedback edge set.
\end{proof}

\begin{example}
  For the graph $G_1(V, \cfgarrow)$ in \pref{fig:cfg1} and WPO $\wpo_1 (V, X,
 \forwardarrow, \backwardarrow)$ in \pref{fig:wpo1}, $B_{\wpo_1} = \{(4,3),
 (8,6), (5,2)\}$. One can verify that $B_{\wpo_1}$ is a feedback edge set for
 $G_1$. \qef
\end{example}

\begin{theorem}
  \label{thm:partial}
  For graph $G(V, \cfgarrow)$ and its WPO $\wpo(V, X, \forwardarrow,
  \backwardarrow)$, $(\cfgarrow \setminus B_\wpo)^+ \subseteq
  \forwardarrowplus$.
\end{theorem}
\begin{proof}
  Each edge $(u,v) \in B_\wpo$ satisfies~\ref{it:Wbd}-(ii) by
  definition. Therefore, all edges in $(\cfgarrow \setminus B_\wpo)$
  must satisfy~\ref{it:Wbd}-(i). Thus, $u \forwardarrowplus v$ for all edges
  $(u,v) \in (\cfgarrow \setminus B_\wpo)$, and $(\cfgarrow \setminus B_\wpo)^*
  \subseteq \forwardarrowstar$.
\end{proof}

Given the tuple $\wpo(V, X, \forwardarrow, \backwardarrow)$ and a set $S$, 
we use $\wpo_{\downharpoonright S}$ to denote 
the tuple $(V \cap S, X \cap S, \forwardarrow_{\downharpoonright S},
\backwardarrow_{\downharpoonright S})$.
The following two theorems enable us to decompose a WPO into sub-WPOs enabling
the use of structural induction when proving properties of WPOs.

\begin{theorem}
  \label{thm:inductdag}
  For graph $G(V, \cfgarrow)$ and its WPO $\wpo(V, X, \forwardarrow,
  \backwardarrow)$, $\wpo_{\downharpoonright C}$ is a WPO for subgraph
  $G_{\downharpoonright C}$ for all $C \in \components_\wpo$.
\end{theorem}
\begin{proof}
  We show that $\wpo_{\downharpoonright C}$ satisfies all conditions
  \pref{it:Wx}--\pref{it:Wbd} in \pref{def:WPO}
  for all $C \in \components_\wpo$.
  Conditions \pref{it:Wx}, \pref{it:Wred}, \pref{it:Wblue}, \pref{it:Whpo}-[\pref{it:Hpartial}, \pref{it:Hoto}, \pref{it:Hdir}, \pref{it:Hnest}] trivially holds true.

  [\ref{it:Whpo}-\ref{it:Hexit}]~If $v \notin C$, \ref{it:Hexit} is true because $(u, v)
  \notin \forwardarrowstar_{\downharpoonright C}$. Else, \ref{it:Hexit} is still
  satisfied with $\forwardarrowstar_{\downharpoonright C}$.

  [\ref{it:Wbd}]
  We show that $u \forwardarrowplus v$ implies $u
  \forwardarrow_{\downharpoonright C}^+ v$, if $u, v \in C$. Let $C =
  \component{h}{x}_{\forwardarrowstar}$ with $x \backwardarrow h$. If 
  $u \forwardarrowplus_{\downharpoonright C} v$ is false, then there exists $w \in
  \postset{u}_{\subforwardarrowstar} \cap \preset{v}_{\subforwardarrowstar}$
  such that $w \notin C$.
  However, $u \forwardarrowplus w$ and $w \notin C$ implies $x \forwardarrowplus
  w$ (using~\ref{it:Hexit}). This contradicts that $(V \cup X, \forwardarrowstar)$ is a
  partial order, because $w \in \preset{v}_{\subforwardarrowstar}$ and $v \in
  \component{h}{x}_{\subforwardarrowstar}$ implies $w \forwardarrowplus x$.
  Thus, \ref{it:Wbd} is satisfied.
\end{proof}

\begin{theorem}
  \label{thm:inductcomp}
  For graph $G(V, \cfgarrow)$ and its WPO $\wpo(V, X, \forwardarrow,
\backwardarrow)$, if $V \cup X = \component{h}{x}_{\forwardarrowstar}$
  for some $(x, h) \in \backwardarrow$,
  then $\wpo_{\downharpoonright S}$ is
a WPO for subgraph $G_{\downharpoonright S}$, where $S \eqdef V \cup X
\setminus \{h, x\}$.
\end{theorem}
\begin{proof}
  We show that $\wpo_{\downharpoonright S}$ satisfies all conditions
  \pref{it:Wx}--\pref{it:Wbd} in \pref{def:WPO}.
  Conditions \pref{it:Wx}, \pref{it:Wred}, \pref{it:Wblue}, \pref{it:Whpo}-[\pref{it:Hpartial}, \pref{it:Hoto}, \pref{it:Hdir}, \pref{it:Hnest}] trivially holds true.

  [\ref{it:Whpo}-\ref{it:Hexit}]~$x$ has no outgoing and $h$ has no incoming
  scheduling constraints. Thus, $\wpo_{\downharpoonright S}$ still
  satisfies~\ref{it:Hexit}.

  [\ref{it:Wbd}]
  Case (i) is still satisfied because $h$ only had outgoing \forward{}s
  and $x$ only had incoming \forward{}s.
  Case (ii) is still satisfied due to \ref{it:Hoto} and \ref{it:Hnest}.
\end{proof}

\begin{example}
  The decomposition of WPO $\wpo_1$ for graph $G_1$ in \pref{fig:cfg-wpo1} is:
  \begin{center}
    \begin{tikzpicture}[auto,node distance=1cm,font=\small]
      \tikzstyle{every node} = [circle, draw, inner sep=0pt,minimum size=2.5ex]
      \node (1) at (0, 0) {$1$};
      \node [right of=1] (2) {$2$};
      \node [right of=2] (3) {$3$};
      \node [right of=3] (4) {$4$};
      \node [right of=4] (3') {$x_3$};
      \node [right of=3', node distance=1.8cm] (5) {$5$};
      \node [right of=5] (2') {$x_2$};
      \node [below of=3, node distance=1cm] (6) {$6$};
      \node [right of=6] (7) {$7$};
      \node [right of=7] (8) {$8$};
      \node [below of=7, node distance=0.42cm] (9) {$9$};
      \node [right of=8] (6') {$x_6$};
      \node [right of=2'] (10) {$10$};
      
      \path (1) edge[forward] (2); 
      \path (2) edge[forward] (3); 
      \path (2) edge[forward] (6); 
      \path (3) edge[forward] (4); 
      \path (4) edge[forward] (3'); 
      \path (3') edge[forward] (5); 
      \path (5) edge[forward] (2'); 
      \path (6) edge[forward] (7); 
      \path (6) edge[forward] (9); 
      \path (7) edge[forward] (8); 
      \path (9) edge[forward] (8); 
      \path (8) edge[forward] (6'); 
      \path (6') edge[forward] (5); 
      \path (2') edge[forward] (10); 
      \draw[thick] ($(7.north west) + (-0.1, 0.1)$) rectangle ($(8.south east) + (0.1, -0.50)$);
      \path (6') edge[backward, bend right=25] (6); 
      \draw[dotted, thick] ($(6.north west) + (-0.1, 0.35)$) rectangle ($(6'.south east) + (0.1, -0.55)$);
      \draw[thick] ($(4.north west) + (-0.1, 0.1)$) rectangle ($(4.south east) + (0.1, -0.1)$);
      \path (3') edge[backward, bend right=25] (3); 
      \draw[dotted, thick] ($(3.north west) + (-0.1, 0.2)$) rectangle ($(3'.south east) + (0.1, -0.15)$);
      \draw[thick] ($(3.north west) + (-0.15, 0.25)$) rectangle ($(5.south east) + (0.15, -0.2)$);
      \path (2') edge[backward, bend right=30] (2); 
      \draw[dotted, thick] ($(2.north west) + (-0.1, 0.8)$) rectangle ($(2'.south east) + (0.1, -0.25)$);
      \draw[thick] ($(1.north west) + (-0.1, 0.85)$) rectangle ($(10.south east) + (0.1, -1.6)$);
      \node[draw=none, left] at ([yshift=-1mm]$(1.north west) + (-0.1, 0.85)$) {$\wpo_{1}$};
      \node[draw=none, left] at ([yshift=-2mm]$(2.north west) + (-0.1, 0.8)$) {$\wpo_{\downharpoonright C_2}$};
      \node[draw=none, above] at ([xshift=2cm]$(3.north west) + (-0.15, 0.2)$) {$\wpo_{\downharpoonright S_2}$};
      \node[draw=none, right] at ([yshift=-1.5mm]$(3'.north east) + (0.1, 0.2)$) {$\wpo_{\downharpoonright C_3}$};
      \node[draw=none, left] at ([yshift=-1cm]$(6.north west) + (-0.1, 0.35)$) {$\wpo_{\downharpoonright C_6}$};
      \node[draw=none, right] at ([yshift=2mm] $(8.south east) + (0.1, -0.5)$) {$\wpo_{\downharpoonright S_6}$};
    \end{tikzpicture}
    \end{center}
  $\components_{\wpo} = \{C_2, C_3, C_6\}$, where
  $C_2 = \{2, 3, 4, x_3, 5, x_2\}$,
  $C_3 = \{3, 4, x_3\}$,
  and $C_6 = \{6, 7, 8, 9, x_6\}$.
 As proved in \pref{thm:inductdag},  
 $\wpo_{\downharpoonright C_2}$,
 $\wpo_{\downharpoonright C_3}$, and $\wpo_{\downharpoonright C_6}$ (shown using dotted lines)
 are WPOs for the subgraphs $G_{\downharpoonright C_2}$,
 $G_{\downharpoonright C_3}$, and $G_{\downharpoonright C_6}$, respectively.
  Furthermore, \pref{thm:inductcomp} is applicable to each of these WPOs.
  Therefore,  $\wpo_{\downharpoonright S_2}$,
  $\wpo_{\downharpoonright S_3}$, and $\wpo_{\downharpoonright S_6}$ (shown using solid lines)
  are WPOs 
  for subgraphs $G_{\downharpoonright S_2}$,
  $G_{\downharpoonright S_3}$, and $G_{\downharpoonright S_6}$, respectively,
  where $S_h = C_h \setminus \{h, x_h\}$ for $h \in \{2, 3, 6\}$.
  For example, $S_6 = C_6 \setminus \{6, x_6\} = \{7, 8, 9\}$.
  Note that 
  $\wpo_{\downharpoonright C_6}$ is a WPO for subgraph
  \raisebox{-11pt}{
  \begin{tikzpicture}[auto,node distance=.7cm,font=\small]
    \tikzstyle{every node} = [rectangle, draw, inner sep=0pt,minimum size=2ex]
    \node  (6) at (0,0) {$6$};
    \node [right of=6] (7) {$7$};
    \node [right of=7] (8)  {$8$};
    \node [below of=7,node distance=0.35cm] (9) {$9$};

    \path (6) edge[cfgedge] (7); 
    \path (6) edge[cfgedge] (9); 
    \path (7) edge[cfgedge] (8); 
    \path (9) edge[cfgedge] (8); 
    \path (8) edge[cfgedge, bend right=24] (6); 
  \end{tikzpicture}},
  while $\wpo_{\downharpoonright S_6}$ is a WPO for subgraph
  \raisebox{-9pt}{
  \begin{tikzpicture}[auto,node distance=.7cm,font=\small]
    \tikzstyle{every node} = [rectangle, draw, inner sep=0pt,minimum size=2ex]
    \node (7) at (0,0) {$7$};
    \node [right of=7] (8)  {$8$};
    \node [below of=7,node distance=0.35cm] (9) {$9$};

    \path (7) edge[cfgedge] (8); 
    \path (9) edge[cfgedge] (8); 
  \end{tikzpicture}}.
  \qef
\end{example}

\begin{definition}
  \label{def:MaximalComponent}
  Given a WPO $\wpo$, $C \in
  \components_{\wpo}$ is a \emph{maximal component} if there does not exist another
  component $C' \in \components_{\wpo}$ such that $C \subset C'$.  
  $\maximalcomponents_\wpo$ denotes the set of maximal components of $\wpo$.
  \qef
\end{definition}
\begin{theorem}
  \label{thm:nocycle}
  For graph $G(V, \cfgarrow)$ and its WPO $\wpo(V, X,
  \forwardarrow, \backwardarrow)$,
  if there is a cycle in $G$ consisting of vertices $V'$,
  then there exists $C \in \maximalcomponents_\wpo$
  such that $V' \subseteq C$.
\end{theorem}
\begin{proof}
  Assume that the theorem is false. Then, there exists multiple maximal components
  that partition the vertices in the cycle. Let $(u,v)$ be an edge in the cycle
  where $u$ and $v$ are in different maximal components. By
  \pref{it:Wbd}, $u \forwardarrowplus v$, and by \pref{it:Hexit}, $x_u
  \forwardarrowplus v$, where $x_u$ is the exit of the maximal component that
  contains $u$. By the definition of the component, $v \forwardarrowplus x_v$, where
  $x_v$ is the exit of the maximal component that contains $v$. Therefore,
  $x_u \forwardarrowplus x_v$. Applying the same reasoning for all such
  edges in the cycle, we get $x_u \forwardarrowplus x_v \forwardarrowplus \dots
  \forwardarrowplus x_u$. This contradicts the fact that $(V \cup X, \forwardarrowstar)$
  is a partial order for the WPO $\wpo$.
\end{proof}

\begin{corollary}
  \label{cor:nocycle}
  For $G(V, \cfgarrow)$ and its WPO $\wpo(V, X,
  \forwardarrow, \backwardarrow)$, 
  if $G$ is a non-trivial strongly connected graph, then 
  there exists $h \in V$ and $x \in X$ such that 
  $\maximalcomponents_{\wpo} = \{\component{h}{x}_{\forwardarrowstar}\}$
  and $\component{h}{x}_{\forwardarrowstar} = V \cup X$.
\end{corollary}
\begin{proof}
  Because there exists a cycle in the graph, there must exists at least one component in the WPO.
  Let $h \in V$ and $x \in X$ be the head and exit of a maximal component in $\wpo$.
  Because $V \cup X$ contains all elements in the WPO,
  $\component{h}{x}_{\forwardarrowstar} \subseteq V \cup X$.
  Now, suppose $\component{h}{x}_{\forwardarrowstar} \nsupseteq V \cup X$.
  Then, there exists $v \in V \cup X$ such that $v \notin \component{h}{x}_{\forwardarrowstar}$.
  If $v \in V$, then there exists a cycle
  $v \cfgarrowplus h \cfgarrowplus v$, because the graph is strongly connected.
  Then, by \pref{thm:nocycle}, $v \in \component{h}{x}_{\forwardarrowstar}$, which is a contradiction.
  If $v \in X$, then there exists $w \in V$ such that $v \backwardarrow w$ by \pref{it:Wred}.
  Due to \pref{it:Hnest}, $w \notin \component{h}{x}_{\forwardarrowstar}$.
  By the same reasoning as the previous case, this leads to a contradiction.
\end{proof}

\section{Deterministic Concurrent Fixpoint Algorithm}
\label{sec:fixpoint}

This section describes a deterministic concurrent algorithm for computing
a fixpoint of an equation system. Given the equation system $\Value =
F(\Value)$ with dependency graph $G(V, \cfgarrow)$, we first construct a WPO
$\wpo (V, X, \forwardarrow, \backwardarrow)$. The algorithm in
\pref{fig:RulesWPO} uses $\wpo$ to compute the fixpoint of $\Value =
F(\Value)$. It defines a concurrent iteration strategy for a WPO: equations are
applied concurrently by following the \forward{}s $\forwardarrow$, while
\backward{}s $\backwardarrow$ act as ``iterate until stabilization'' operators,
checking the stabilization at the exits and iterating the components.

\newcommand{\appf}{\texttt{Apply$F$}}
\newcommand{\cstab}{\texttt{ComponentStabilized}}
\newcommand{\pred}{\texttt{NumSchedPreds}}
\newcommand{\cpred}{\texttt{NumOuterSchedPreds}}
\newcommand{\resetcd}{\texttt{Set$\Map$ForComponent}}
\begin{figure}[t]
  \begin{subfigure}[b]{\linewidth}
  {\small
  \begin{mathpar}
    \inferrule*[right=$\InitRule$]
    { }
    { \textrm{forall } v \in V, \Value[v \mapsto (\textrm{$v$ is entry ? } \top  :  \bot)] \\
      \textrm{forall } v \in V \cup X, \Map[v \mapsto 0]
    }
    \and
    \inferrule*[right=$\NERule$]
    { v \in V \\ \Map(v) = \pred{}(v) }
    { \Map[v \mapsto 0] \\ \appf{}(v) \\ \textrm{forall } v \forwardarrow w, \Map[w \mapsto (\Map[w] + 1)] }
    \and
    \inferrule*[right=$\ECRule$]
    { x \in X \\ \Map(x) = \pred{}(x) \\ \cstab{}(x) }
    { \Map[x \mapsto 0] \\ \textrm{forall } x \forwardarrow w, \Map[w \mapsto (\Map(w) + 1)] }
    \and 
    \inferrule*[right=$\ENCRule$]
    { x \in X \\ \Map(x) = \pred{}(x) \\ \neg\cstab{}(x) }
    { \Map[x \mapsto 0] \\ \resetcd{}(x) }
  \end{mathpar}
  }%
  \end{subfigure}
  \begin{subfigure}[b]{\linewidth}
    {\small
      \begingroup 
      \setlength{\belowdisplayskip}{0pt}%
      \setlength{\belowdisplayshortskip}{0pt}%
      \begin{align*}
        \appf{}(v) &\eqdef \Value[v \mapsto \big(v \in \text{ image of } \backwardarrow \textrm{ ? } \Value(v) \widen F_{v}(\Value) \textrm{ : } F_{v}(\Value)\big)] \\
        \cstab{}(x) &\eqdef \exists h \in V. x \backwardarrow h \wedge F_{h}(\Value) \sqsubseteq \Value(h) \\
        \resetcd{}(x) &\eqdef \textrm{forall } v \in C_x, \Map[v \mapsto \cpred{}(v, x)] \\
        \pred{}(v) &\eqdef |\{ u \in V \cup X \mid u \forwardarrow v \}| \\
        \cpred{}(v, x) &\eqdef |\{ u \in V \cup X \mid u \forwardarrow v, u \notin C_x, v \in C_x\}|
      \end{align*}
      \endgroup %
    }%
  \end{subfigure}
  \caption{
    Deterministic concurrent fixpoint algorithm for WPO.
    $\Value$ maps an element in $V$ to its value.
    $\Map$ maps an element in $V \cup X$ to its count of executed scheduling predecessors.
    Operations on $\Map$ are atomic.}
  \label{fig:RulesWPO}
  \vspace{-2ex}
\end{figure}

Except for
the initialization rule \InitRule{}, which is applied once at the beginning,
rules in \pref{fig:RulesWPO} are applied concurrently whenever some element in $ V \cup X$ satisfies
the conditions. The algorithm uses a value map $\Value$, which maps an element
in $V$ to its abstract value, and a count map $\Map$, which maps an element in
$V \cup X$ to its counts of executed scheduling predecessors. Access to the value map
$\Value$ is synchronized by \forward{}s, and operations on $\Map$ are assumed to
be atomic. Rule \InitRule{} initializes values for elements in $V$ to $\bot$
except for the entry of the graph, whose value is initialized to $\top$. The
counts for elements in $V \cup X$ are all initialized to 0.

Rule \NERule{} applies to a non-exit element $v \in V$ whose scheduling predecessors
are all executed ($\Map(v) = \pred(v)$). This rule applies the
function $F_v$ to update the value $\Value_v$ ($\appf(v)$). Definition of the
function $\appf$ shows that the widening is applied at the image of
$\backwardarrow$ (see \pref{thm:widening}). The rule then notifies the
scheduling successors of $v$ that $v$ has executed by incrementing their counts.
Because elements within a component can be iterated multiple times, the
count of an element is reset after its execution. If there is no component in
the WPO, then only the \NERule{} rule is applicable, and the algorithm reduces to a
DAG scheduling algorithm.

Rules \ECRule{} and \ENCRule{} are applied to an exit $x$ ($x \in X$)
whose scheduling predecessors are all executed ($\Map(x) = \pred(x)$).
If the component $C_x$ is stabilized, \ECRule{} is applied, and \ENCRule{} otherwise.
A component is \emph{stabilized} if iterating it once more does not change the values
of elements inside the component.
Boolean function \cstab{} checks the stabilization of $C_x$ by checking the stabilization of its head (see \pref{thm:convergence}).
Upon stabilization, rule \ECRule{} notifies the scheduling successors of $x$
and resets the count for $x$. 

\begin{example}
  An iteration sequence generated by rules in \pref{fig:RulesWPO} for WPO $\wpo_1$ (\pref{fig:wpo1}) is:

  \newcommand{\PreserveBackslash}[1]{\let\temp=\\#1\let\\=\temp}
\newcolumntype{C}[1]{>{\PreserveBackslash\centering}p{#1}}
\newcolumntype{R}[1]{>{\PreserveBackslash\raggedleft}p{#1}}
\newcolumntype{L}[1]{>{\PreserveBackslash\raggedright}p{#1}}
\begin{center}
{\small
\begin{tabular}{c|cccccccccccccc}
  & \multicolumn{5}{l}{Time step in $\mathbb{N}$ $\longrightarrow$} \\
\midrule
  Scheduled element    & 1 & 2  & 3 & 4 & $x_3$ & 3     & 4 & $x_3$ & 3 & 4     & $x_3$ & 5 & $x_2$ & 10 \\
    $u \in V \cup X$   &   &    & 6 & 7 & 8     & $x_6$ & 6 & 7     & 8 & $x_6$ \\
                       &   &    &   & 9 &       &       &   & 9        \\
\end{tabular}
}
\end{center}

  The initial value of $\Map(8)$ is 0. Applying \NERule{} to $7$ and
  $9$ increments $\Map(8)$ to 2. $\Map(8)$ now equals $\pred(8)$, and \NERule{}
  is applied to $8$.
  Applying \NERule{} to $8$ updates $\Value_8$ by applying the function $F_8$,
  increments $\Map(x_6)$, and resets $\Map(8)$ to 0.
  Due to the reset, same thing happens when $C_6$ is iterated once more.
  The initial value of $\Map(x_6)$ is $0$. Applying \NERule{} to $8$ increments
  $\Map(x_6)$ to 1, which equals $\pred(x_6)$. The stabilization of component
  $C_6$ is checked at $x_6$. If it is stabilized, \ECRule{} is applied to $x_6$,
  which increments $\Map(5)$ and resets $\Map(x_6)$ to~$0$.
  \qef
\end{example}

The rule \ENCRule{} is applied if the component $C_x$ is not stabilized.
Unlike rule \ECRule{}, \ENCRule{} does not notify the scheduling successors of $x$. 
It resets the counts for each element in $C_x$ as well as that for $x$.
$\resetcd{}(x)$ sets the count for $v \in C_x$ to $\cpred(v, x)$, which is the number of its
scheduling predecessors not in $C_x$.
In particular, the count for the head of
$C_x$, whose scheduling predecessors are all not in $C_x$, is set to the number
of all scheduling predecessors, allowing rule \NERule{} to be applied to the
head.
The map $\cpred(v, x)$ can be computed during WPO construction.

\begin{example}
  \label{exa:IrreducibleGraph}
  Let $G_2$ be 
   \raisebox{-12pt}{
  \begin{tikzpicture}[auto,node distance=.7cm,font=\small]
    \tikzstyle{every node} = [rectangle, draw, inner sep=0pt,minimum size=2.5ex]
    \node (1) at (0,0) {$1$};
    \node [right of=1] (2) {$2$};
    \node [right of=2] (3) {$3$};
    \node [right of=3] (4) {$4$};
    \node [above right of=2, node distance=.6cm] (5) {$5$};
    \node [below right of=1, node distance=.6cm] (6) {$6$};

    \path (1) edge[cfgedge] (6.west); 
    \path (1) edge[cfgedge] (2); 
    \path (2) edge[cfgedge] (3); 
    \path (2.north) edge[cfgedge] (5.west); 
    \path (3) edge[cfgedge] (4); 
    \path (4) edge[cfgedge, bend right=20] (3); 
    \path (6) edge[cfgedge] (4.south); 
    \path (5) edge[cfgedge] (4.north); 
    \path (3) edge[cfgedge, bend right=20] (2); 
  \end{tikzpicture}
  }%
  and its WPO $\wpo_2$ be
  \raisebox{-12pt}{
  \begin{tikzpicture}[auto,node distance=.7cm,font=\small]
    \tikzstyle{every node} = [circle, draw, inner sep=0pt,minimum size=2.5ex]
    \node (1) {$1$};
    \node [right of=1] (2) {$2$};
    \node [right of=2] (3) {$3$};
    \node [above right of=2, node distance=.6cm] (5) {$5$};
    \node [right of=3] (4) {$4$};
    \node [right of=4] (3') {$x_3$};
    \node [right of=3'] (2') {$x_{2}$};
    \node [below right of=1, node distance=.6cm] (6) {$6$};
    
    \path (1) edge[forward] (2); 
    \path (1) edge[forward] (6); 
    \path (2) edge[forward] (3); 
    \path (2) edge[forward] (5); 
    \path (3) edge[forward] (4); 
    \path (4) edge[forward] (3'); 
    \path (3') edge[forward] (2'); 
    \path (5) edge[forward] (4); 
    \path (6) edge[forward] (4.south); 
    \path (3') edge[backward, bend left=30] (3); 
    \path (2') edge[backward, bend left=30] (2); 
  \end{tikzpicture}
  }.
  An iteration sequence generated by the concurrent fixpoint algorithm for WPO $\wpo_2$ is:

  \begin{center}
{\small
\begin{tabular}{c|cccccccccccccc}
  & \multicolumn{5}{l}{Time step in $\mathbb{N}$ $\longrightarrow$} \\
\midrule
  Scheduled element  & 1 & 2 & 3 & 4 & $x_3$ & 3 & 4 & $x_3$ & $x_2$ & 2 & 3 & 4 & $x_3$ & $x_2$  \\
   $u \in V \cup X$  &   & 6 & 5 &   &       &   &   &       &       &   & 5  \\
\end{tabular}
}
\end{center}

  Consider the element $4$, whose scheduling predecessors in the WPO are $3$,
  $5$, and $6$. Furthermore, $4 \in C_3$ and $4 \in C_2$ with $C_3 \subsetneq
  C_2$. After \NERule{} is applied to $4$, $\Map(4)$ is reset to 0. Then, if the
  stabilization check of $C_3$ fails at $x_3$, \ENCRule{} sets $\Map(4)$ to 2,
  which is $\cpred(4, x_3)$. If it is not set to 2, then the fact that elements $5$
  and $6$ are executed will not be reflected in $\Map(4)$, and the iteration over
  $C_3$ will be blocked at element $4$. If the stabilization check of
  $C_2$ fails at $x_2$, \ENCRule{} sets $\Map(4)$ to $\cpred(4,
  x_2) = 1$. %
  \qef
\end{example}

In \appf{}, the image of $\backwardarrow$ is chosen as the set of
widening points. These are heads of the components. The following theorem proves
that the set of component heads is an admissible set of the widening points,
which guarantee the termination of the fixpoint computation:

\begin{theorem}
  \label{thm:widening}
  Given a dependency graph $G(V, \cfgarrow)$ and its WPO $\wpo(V, X,
  \forwardarrow, \backwardarrow)$, the set of component heads is an
  admissible set of widening points.
\end{theorem}
\begin{proof}
  \pref{thm:feedback} proves that $B_\wpo \eqdef \{ (u, v) \in \cfgarrow \mid
  \exists x \in X. u \in \component{v}{x}_{\forwardarrowstar} \land  x
  \backwardarrow v\}$ is a feedback edge set. Consequently, the set of component
  heads $\{ h \mid \exists x\in X. x \backwardarrow h\}$ is a feedback
  vertex set. Therefore, the set $W$ is an admissible set of widening points
  \cite{kn:CC77}.
\end{proof}

\begin{example}
  The set of component heads $\{2, 3, 6\}$ is an admissible set of widening
  points for the WPO $\wpo_1$ in \pref{fig:wpo1}.
  \qef
\end{example}

The following theorem justifies our definition of \cstab{}; viz., checking the
stabilization of $\Value_h$ is sufficient for checking the stabilization of the
component $C_h$.

\begin{theorem}
  \label{thm:convergence}
  During the execution of concurrent fixpoint algorithm with WPO $\wpo(V, X, \forwardarrow, \backwardarrow)$,
  stabilization of the head $h$ implies the stabilization of the component $C_h$ at its exit
  for all $C_h \in \components_{\wpo}$.
\end{theorem}
\begin{proof}
  Suppose that there exists an element $v \in C_h$ such that $v$ is not stabilized
  although the head $h$ is stabilized.
  That is, suppose that $\Value_v$ changes when $C_h$ is iterated once more.
  For this to be possible, there must exist $u \in V$ such that $u \cfgarrow v$ and
  $\Value_u$ changed after the last update of $\Value_v$.
  By $u \cfgarrow v$ and \pref{it:Wbd}, either (i)~$u \forwardarrowplus v$ or 
  (ii)~$u \in \component{v}{x}_{\forwardarrowstar} = C_v$ and $x \backwardarrow v$ for some $x \in X$.
  It cannot be the case~(i), because $u \forwardarrowplus v$ and \pref{it:Hexit}
  imply that $\Value_u$ cannot be updated after the last update of $\Value_v$.
  Therefore, it should always be the case~(ii). By \pref{it:Hnest}, $C_v \subsetneq C_h$.
  However, because $u \in C_v$, our algorithm checks the stabilization of $v$ at
  the exit of $C_v$
  after the last update of $\Value_u$. This contradicts the assumption that
  $\Value_u$ changed after the last update of $\Value_v$. 
\end{proof}

A WPO $\wpo(V, X, \forwardarrow, \backwardarrow)$ where
$V = \{v\}$, $X = \forwardarrow = \backwardarrow = \emptyset$
is said to be a \emph{trivial WPO}, which is represented as
\raisebox{-2.5pt}{
  \begin{tikzpicture}[auto,node distance=.7cm,font=\small]
    \tikzstyle{every node} = [circle, draw, inner sep=0pt,minimum size=2.5ex]
    \node (v) {$\small v$};
  \end{tikzpicture}
}. It can only be a WPO for a trivial SCC with vertex $v$.
A WPO $\wpo(V, X, \forwardarrow, \backwardarrow)$ where
$V = \{h\}$, $X = \{x\}$, $\forwardarrow = \{(h, x)\}$, $\backwardarrow = (x, h)$
is said to be a \emph{self-loop WPO}, and is represented as
\raisebox{-2pt}{
  \begin{tikzpicture}[auto,node distance=.7cm,font=\small]
    \tikzstyle{every node} = [circle, draw, inner sep=0pt,minimum size=2.5ex]
    \node (h) {$\small h$};
    \node [right of=h] (x) {$\small x$};
    \path (h) edge[forward] (x); 
    \path (x) edge[backward, bend right=30] (h); 
  \end{tikzpicture}
}. It can only be a WPO for a trivial SCC with vertex $h$ or a single vertex $h$ with a self-loop.

The following theorem proves that the concurrent fixpoint algorithm in
\pref{fig:RulesWPO} is deterministic.

\begin{theorem}
  \label{thm:deter}
  Given a WPO $\wpo(V, X, \forwardarrow, \backwardarrow)$ for a graph $G(V, \cfgarrow)$
  and a set of monotonone, deterministic functions $\{F_v \mid v \in V\}$,
  concurrent fixpoint algorithm in \pref{fig:RulesWPO} is deterministic,
  computing the same approximation of the least fixpoint for the given set of functions.
\end{theorem}
\begin{proof} We use structural induction on the WPO $\wpo$ to show this.

  \noindent \textsc{[Base case]:} 
  The two cases for the base case are (i)~
  $\wpo = $
  \raisebox{-2.5pt}{
    \begin{tikzpicture}[auto,node distance=.7cm,font=\small]
      \tikzstyle{every node} = [circle, draw, inner sep=0pt,minimum size=2.5ex]
      \node (v) {$\small v$};
    \end{tikzpicture}
  }
  and (ii)~
  $\wpo = $
  \raisebox{-2pt}{
    \begin{tikzpicture}[auto,node distance=.7cm,font=\small]
      \tikzstyle{every node} = [circle, draw, inner sep=0pt,minimum size=2.5ex]
      \node (h) {$\small h$};
      \node [right of=h] (x) {$\small x$};
      \path (h) edge[forward] (x); 
      \path (x) edge[backward, bend right=30] (h); 
    \end{tikzpicture}
  }.
  If
  $\wpo = $
  \raisebox{-2.5pt}{
    \begin{tikzpicture}[auto,node distance=.7cm,font=\small]
      \tikzstyle{every node} = [circle, draw, inner sep=0pt,minimum size=2.5ex]
      \node (v) {$\small v$};
    \end{tikzpicture}
  },
  then $v$ is the only vertex in $G$.
  Functions are assumed to be deterministic, so
  applying the function $F_v()$ in rule $\NERule$ of \pref{fig:RulesWPO}
  is deterministic. Because $F_v()$ does not take any arguments, the computed value $\Value_v$ is a unique fixpoint of $F_v()$.

  If
  $\wpo = $
  \raisebox{-2pt}{
    \begin{tikzpicture}[auto,node distance=.7cm,font=\small]
      \tikzstyle{every node} = [circle, draw, inner sep=0pt,minimum size=2.5ex]
      \node (h) {$\small h$};
      \node [right of=h] (x) {$\small x$};
      \path (h) edge[forward] (x); 
      \path (x) edge[backward, bend right=30] (h); 
    \end{tikzpicture}
  },
  then $h$ is the only vertex in $G$.
  If $h$ has a self-loop, function $F_h(\Value_h)$ and widening operator $\widen$ may need to be applied multiple
  times to reach a post-fixpoint (approximation of the least fixpoint) of $F_h(\Value_h)$.
  Rule $\NERule$ applies them once on $\Value_h$ and signals the exit $x$.
  If the post-fixpoint of $F_h(\Value_h)$ is not reached, $\cstab$ returns false, and rule $\ENCRule$
  in \pref{fig:RulesWPO} applies rule $\NERule$ on $h$ again.
  If the post-fixpoint is reached, $\cstab$ returns true, and rule $\ECRule$
  in \pref{fig:RulesWPO} stops the algorithm.
  Because $F_h(\Value_h)$ and $\widen$ are deterministic, each iteration is deterministic,
  and the entire sequence of iterations are deterministic.
  The computed value $\Value_h$ is a post-fixpoint of $F_h(\Value_h)$.

  \noindent \textsc{[Inductive step]:}
  By condition \pref{it:Hnest} and \pref{thm:inductdag},
  $\wpo$ can be decomposed into a set of WPOs of its maximal components and trivial WPOs.
  The two cases for the inductive step are (i)~the decomposition of $\wpo$
  is a single WPO of the maximal component
  and (ii)~$\wpo$
  is decomposed into multiple WPOs.

  If the decomposition of $\wpo$ is a single WPO of the maximal component
  $\component{h}{x}_{\forwardarrowstar}$,
  then by \pref{thm:inductcomp},
  $\wpo = $
  \raisebox{-6pt}{
    \begin{tikzpicture}[auto,node distance=.7cm,font=\small]
      \tikzstyle{every node} = [circle, draw, inner sep=0pt,minimum size=2.5ex]
      \node (h) at (0, 0) {$\small h$};
      \node [rectangle, thick, minimum size=4ex, right of=h, node distance=1cm] (w) {$\ \wpo_{\downharpoonright S}\ $};
      \node [right of=w, node distance=1cm] (x) {$\small x$};

      \path (h) edge[forward] ([yshift=-1.5mm] w.north west); 
      \path (h) edge[forward] (w); 
      \path (h) edge[forward] ([yshift=1.5mm] w.south west); 
      \path ([yshift=-1.5mm] w.north east) edge[forward] (x); 
      \path (w) edge[forward] (x); 
      \path ([yshift=1.5mm] w.south east) edge[forward] (x); 
      \path (x) edge[backward, bend right=35] (h); 
    \end{tikzpicture}
  },
  where $S \eqdef \component{h}{x}_{\forwardarrowstar} \setminus \{h, x\}$ and $\wpo_{\downharpoonright S}$ is a WPO. By the induction hypothesis,
  the fixpoint algorithm is deterministic for $\wpo_{\downharpoonright S}$. 
  The head $h$ of $\wpo$ is its unique source, so each iteration begins with
  the application of rule $\NERule$ on $h$.
  After applying $F_h(\cdot)$ and $\widen$, rule $\NERule$ signals all its scheduling successors,
  initiating the iteration over $\wpo_{\downharpoonright S}$.
  Because all sinks of $\wpo_{\downharpoonright S}$ are connected to the exit $x$ of $\wpo$,
  $x$ will be processed after the iteration finishes.
  Thus, $\Value_h$ remains fixed during the iteration.
  A single iteration over $\wpo_{\downharpoonright S}$ is identical to running the fixpoint algorithm on
  $\wpo_{\downharpoonright S}$ with the set of functions $\{F_v' \mid v \in V' \}$, where
  $\wpo_{\downharpoonright S}$ is a WPO for subgraph
  $G_{\downharpoonright S}$ and function $F_v'$ is a partial application of $F_v$ that binds
  the argument that take $\Value_h$ to its current value.
  The number of functions and the arity of each function decrease.
  Because $F_h(\cdot)$ and $\widen$ are deterministic, and each iteration over
  $\wpo_{\downharpoonright S}$ is deterministic, it is deterministic for $\wpo$.
  The algorithm iterates until the post-fixpoint of $F_h(\cdot)$ is reached,
  and by \pref{thm:convergence}, 
  the computed value $\Value_v$ is a post-fixpoint of $F_v(\cdot)$ for all $v \in V$.

  If $\wpo$ is decomposed into multiple WPOs,
  then by the induction hypothesis,
  the fixpoint algorithm deterministically computes the post-fixpoints for all sub-WPOs.
  Let $\wpo_i(V_i, X_i, \forwardarrow_{\downharpoonright V_i \cup X_i}, \backwardarrow_{\downharpoonright V_i \cup X_i})$ be an arbitrary sub-WPO.
  For any $u \in V \setminus V_i$ such that $u \cfgarrow v$, we have $u \forwardarrowplus v$
  by \pref{it:Wbd} and \pref{it:Hnest}.
  Hence, $v$ is processed after $u$.
  Combined with \pref{it:Hexit}, $\Value_u$ remains unchanged during the iteration of
  $\wpo_i$.
  Single iteration over $\wpo_i$ is equal to running the fixpoint algorithm on
  $\wpo_i$ with the set of the functions $\{F_v' \mid v \in V'\}$,
  where function $F_v'$ is a partial application of $F_v$ that binds the
  arguments in $V \setminus V_i$ to their current values.
  The number of functions decreases, and the arity of the functions does not increase.
  The outer scheduling predecessors of $W_i$ can be ignored in the iterations by $\resetcd$ in rule $\ENCRule$.
  Therefore, by the induction hypothesis, each iteration over $\wpo_i$ is deterministic,
  and because the choice of $\wpo_i$ is arbitrary,
  the algorithm is deterministic for $\wpo$.
  Furthermore, $(v, u) \notin \cfgarrowplus$ for any $u \in V \setminus V_i$ such that $u \cfgarrow v$,
  because its negation would contradict \pref{thm:nocycle}.
  Therefore, change in $\Value_v$ does not change $\Value_u$, and $\Value_u$ is still
  a post-fixpoint of $F_u(\cdot)$.
\end{proof}

\section{Algorithms for WPO Construction}
\label{sec:Algorithm}

\newcommand{\algtd}{\texttt{ConstructWPO\textsuperscript{TD}}}
\newcommand{\algtdscc}{\texttt{ConstructWPO\textsuperscript{TD}\textsubscript{SCC}}}
\newcommand{\algbu}{\texttt{ConstructWPO\textsuperscript{BU}}}

This section presents two algorithms for constructing a WPO for a graph $G(V, \cfgarrow)$. 
The first algorithm, \algtd{}, is a top-down recursive algorithm that is
inefficient but intuitive~(\pref{sec:AlgorithmTD}).
The second one, \algbu{}, is an efficient bottom-up iterative algorithm that has
almost-linear time complexity~(\pref{sec:AlgorithmBU}).
Both algorithms do not introduce superfluous \forward{}s 
that could restrict the parallelism during the fixpoint computation.

\subsection{Top-down Recursive Construction}
\label{sec:AlgorithmTD}

\pref{alg:td} presents a top-down recursive algorithm \algtd{}, which acts as a
proxy between the axiomatic characterization of WPO in \pref{sec:wpo} and the
efficient construction algorithm \algbu{} in \pref{sec:AlgorithmBU}.

\SetKwProg{Def}{def}{:}{end}
\SetKwFunction{TWPOKw}{trivialWPO}
\SetKwFunction{SWPOKw}{self-loopWPO}
\SetKwFunction{CWPOSCCKw}{sccWPO}

\begin{figure}[t]
  \begin{minipage}[t]{0.96\textwidth}
\SetInd{.3em}{0.5em}
  \vspace{0pt}
  \begin{algorithm}[H]
    \small
    \SetKw{andKw}{and}
    \SetKw{returnKw}{return}
    \SetKw{stKw}{s.t.}
    \DontPrintSemicolon
    \setlength{\columnsep}{15pt}
    \Indmm
    \KwIn{Directed graph $G(V, \cfgarrow)$, Depth-first forest $D$}
    \KwOut{WPO $\wpo(V, X, \forwardarrow, \backwardarrow)$}
    \Indpp
    \begin{multicols}{2}
      \small
      \tikzmk{A}
      $G_1, G_2, \ldots, G_k \coloneqq SCC(G)$\label{li:td-scc}\Comment*[r]{Maximal SCCs.}
      \ForEach{$i \in [1, k]$} {
        \Comment*[l]{WPOs for SCCs.}
        $(V_i, X_i, \forwardarrow_i, \backwardarrow_i), h_i, x_i \coloneqq$ \CWPOSCCKw{$G_i, D$}\label{li:td-rec}
      }
      $V, X, \forwardarrow, \backwardarrow  \coloneqq \bigcup\limits_{i=1}^{k} V_i, \bigcup\limits_{i=1}^{k} X_i, \bigcup\limits_{i=1}^{k} \forwardarrow_i, \bigcup\limits_{i=1}^{k} \backwardarrow_i$\label{li:td-merge}\;
      \Comment*[l]{Edges between different maximal SCCs.}
      \ForEach{$u \cfgarrow v$ \stKw $u \in V_i \land v \in V_j \land i \neq j$}{ \label{li:td-for}
        $\forwardarrow \coloneqq \forwardarrow \cup \{(x_i, v)\}$\label{li:td-forwardarrow}
      }
      \returnKw{$(V, X, \forwardarrow, \backwardarrow)$}
      \BlankLine
      \tikzmk{B} \boxit{gray}{.97}
      \tikzmk{A}
      \Def{\TWPOKw{$\etch$}}{
      \tikzmk{B} \boxit{cyan}{.92}
        \returnKw{$(\{\etch\}, \emptyset, \emptyset, \emptyset), \etch, \etch$}\label{li:twpo-ret}
      }
      \tikzmk{A}
      \Def{\SWPOKw{$\etch$}}{
      \tikzmk{B} \boxit{cyan}{.92}
        $x_\etch \coloneqq$ new exit\label{li:swpo-newvar}\;
        \returnKw{$(\{\etch\}, \{x_\etch\}, \{(\etch, x_\etch)\}, \{(x_\etch, \etch)\}), \etch, x_\etch$}\label{li:swpo-ret}
      }
      \tikzmk{A}
      \Def{\CWPOSCCKw{$G, D$}}{
      \tikzmk{B} \boxit{cyan}{.92}
        \Comment*[l]{$G$ is strongly connected.}
        $\etch \coloneqq \argmin_{v\in V}DFN(D, v)$\label{li:tdscc-etch}\Comment*[r]{Minimum DFN.}
        $B \coloneqq \{ (v, \etch) \mid  v \in V \text{ and } v \cfgarrow \etch\}$\label{li:tdscc-B}\;
        \lIf{$|B| = 0$}{
          \returnKw{\TWPOKw{$\etch$}}\label{li:tdscc-twpo}
        }
        \lIf{$|V| = 1$}{
          \returnKw{\SWPOKw{$\etch$}}\label{li:tdscc-slwpo}
        }
        $x_\etch \coloneqq$ new exit\label{li:tdscc-newvar}\;
        $V' \coloneqq V \cup \{x_\etch \} \setminus \{\etch\}$\label{li:tdscc-removal}\;
        $\cfgarrow' \coloneqq \cfgarrow_{\downharpoonright V'} \cup \{ (v, x_\etch) \mid (v, \etch) \in B\}$\label{li:tdscc-modcfg}\;
        \Comment*[l]{WPO for modified graph.}
        $V', X', \forwardarrow', \backwardarrow' \coloneqq$ \algtd{$((V', \cfgarrow'), D_{\downharpoonright V'})$}\label{li:tdscc-rec}\;
        $X \coloneqq X' \cup \{x_\etch \}$\label{li:tdscc-move}\;
        $\forwardarrow \coloneqq \forwardarrow' \cup \{ (\etch, v) \mid  v \in V \text{ and } \etch \cfgarrow v\}$\label{li:tdscc-forwardarrow}\;
        $\backwardarrow \coloneqq \backwardarrow' \cup \{(x_\etch, \etch)\}$\label{li:tdscc-backwardarrow}\;%
        \returnKw{$(V, X, \forwardarrow, \backwardarrow), \etch, x_\etch$}
      }
    \end{multicols}
    \vspace{-2ex}
    \caption{\algtd{}$(G, D)$}
    \label{alg:td}
  \end{algorithm}
\SetInd{.5em}{1em}  
  \end{minipage}
\end{figure}

\algtd{} is parametrized by the depth-first forest (DFF) of the graph $G$, and it may yield
a different WPO for a different DFF.
\algtd{} begins with the identification of the maximal strongly connected
components (SCCs) in $G$ on \pref{li:td-scc}.
An SCC $G_i$ is maximal if there does not exists another SCC that contains all vertices and edges of $G_i$.
A WPO for an SCC $G_i$ is constructed by a call \CWPOSCCKw{$G_i, D$} on
\pref{li:td-rec}. This call returns a WPO
$(V_i, X_i, \forwardarrow_i, \backwardarrow_i)$, head $h_i$, and exit $x_i$.
In case of trivial SCCs, the head and the exit are assigned
the vertex in $G$. In other cases, the returned value satisfies
$(x_i, h_i) \in \backwardarrow_i$ and $\component{h_i}{x_i}_{\forwardarrow_i^*} = V_i \cup X_i$~(\pref{lem:structure}).
\pref{li:td-merge} initializes the WPO for the graph $G$ to union of the WPOs for the
SCCs, showing the inductive structure mentioned in \pref{thm:inductdag}.
On \pref{li:td-forwardarrow}, \forward{}s are added for the dependencies
that cross the maximal SCCs. $x \forwardarrow v$ is added for a dependency $u \cfgarrow v$, where $x$ is the exit of the maximal component WPO that contains $u$ but not $v$.
This ensures that \pref{it:Wbd} and \pref{it:Hexit} are satisfied.

The function \CWPOSCCKw\ takes as input an SCC and its DFF, and returns a WPO, a head, and an exit for this SCC.
It constructs the WPO by removing the head $\etch$ to break the SCC,
adding the exit $x_\etch$ as a unique sink, using \algtd{} to construct a WPO
for the modified graph, and appending necessary elements for the removed head.
Ignoring the exit $x_\etch$, it shows the inductive structure mentioned in \pref{thm:inductcomp}.
\pref{li:tdscc-etch} chooses a vertex with minimum DFS numbering (DFN) as the head.
Incoming edges to the head $\etch$ are back edges for DFF $D$ on \pref{li:tdscc-B}
because $\etch$ has the minimum DFN and can reach to all other vertices in the graph.
If there are no back edges, then the input SCC is trivial, and a trivial WPO with single
element is returned on \pref{li:tdscc-twpo}.
If there is only one vertex in the SCC (with a self-loop), corresponding self-loop WPO
is returned on \pref{li:tdscc-slwpo}.
For other non-trivial SCCs, $\etch$ is removed from the graph and newly created $x_\etch$ is added as
a unique sink on Lines~\ref{li:tdscc-newvar}--\ref{li:tdscc-modcfg}. 
A call \algtd{}$((V', \cfgarrow'), D_{\downharpoonright V'})$ on \pref{li:tdscc-rec} returns a WPO for the modified
graph. Exit $x_\etch$ is moved from $V'$ to $X'$ on \pref{li:tdscc-move},
scheduling constraints regarding the head $\etch$ is added on \pref{li:tdscc-forwardarrow},
and $x_\etch \backwardarrow \etch$ is added on
\pref{li:tdscc-backwardarrow} to satisfy \pref{it:Wbd} for the removed back
edges.

\begin{example}
  Consider the graph $G_3$ in \pref{fig:cfg3}.
  $SCC(G_3)$ on \pref{li:td-scc} of \algtd{}
  returns a trivial SCC with vertex 1 and
  three non-trivial SCCs with vertex sets $\{5,6,7,8\}$, $\{2,3\}$, and $\{4\}$.
  For the trivial SCC, \CWPOSCCKw\ on \pref{li:td-rec} returns
  (\raisebox{-3pt}{
    \begin{tikzpicture}[auto,node distance=.7cm,font=\small]
      \tikzstyle{every node} = [circle, draw, inner sep=0pt,minimum size=2.5ex]
      \node (1) {$\small 1$};
    \end{tikzpicture}
  }, $1$, $1$).
  For the non-trivial SCCs,
  it returns 
  (\raisebox{-5.5pt}{
    \begin{tikzpicture}[auto,node distance=.7cm,font=\small]
      \tikzstyle{every node} = [circle, draw, inner sep=0pt,minimum size=2.5ex]
      \node (5) {$\small 5$};
      \node [right of=5] (6) {$\small 6$};
      \node [right of=6] (7) {$\small 7$};
      \node [right of=7] (6') {$\small x_6$};
      \node [right of=6'] (8) {$\small 8$};
      \node [right of=8] (5') {$\small x_5$};
      \path (5) edge[forward] (6); 
      \path (6) edge[forward] (7); 
      \path (7) edge[forward] (6'); 
      \path (6') edge[forward] (8); 
      \path (8) edge[forward] (5'); 
      \path (6') edge[backward, bend left=28] (6); 
      \path (5') edge[backward, bend right=20] (5); 
    \end{tikzpicture}
  }, $5$, $x_5$),
  (\raisebox{-3.5pt}{
    \begin{tikzpicture}[auto,node distance=.7cm,font=\small]
      \tikzstyle{every node} = [circle, draw, inner sep=0pt,minimum size=2.5ex]
      \node (2) {$\small 2$};
      \node [right of=2] (3) {$\small 3$};
      \node [right of=3] (2') {$\small x_2$};
      \path (2) edge[forward] (3); 
      \path (3) edge[forward] (2'); 
      \path (2') edge[backward, bend right=25] (2); 
    \end{tikzpicture}
  }, $2$, $x_2$), and
  (\raisebox{-3.5pt}{
    \begin{tikzpicture}[auto,node distance=.7cm,font=\small]
      \tikzstyle{every node} = [circle, draw, inner sep=0pt,minimum size=2.5ex]
      \node (4) {$\small 4$};
      \node [right of=4] (4') {$\small x_4$};
      \path (4) edge[forward] (4'); 
      \path (4') edge[backward, bend right=30] (4); 
    \end{tikzpicture}
  }, $4$, $x_4$).
  On \pref{li:td-forwardarrow},
  \forward{}s
  $1 \forwardarrow 2$, $1 \forwardarrow 5$, $x_5 \forwardarrow 3$,
  $x_5 \forwardarrow 4$, and $x_2 \forwardarrow 4$ are added for
  the edges 
  $1 \cfgarrow 2$, $1 \cfgarrow 5$, $7 \cfgarrow 3$, $8 \cfgarrow 4$,
  and $3 \cfgarrow 4$, respectively.
  The final result is identical to $\wpo_3$ in \pref{fig:wpo3}.

  Now, consider the execution of \CWPOSCCKw\ when given the SCC with vertex set $V =
  \{5, 6, 7, 8\}$ as input. On \pref{li:tdscc-etch}, the vertex $5$ is chosen as
  the head $\etch$, because it has the minimum DFN among $V$. The set $B$ on
  \pref{li:tdscc-B} is $B = \{(8,5)\}$. The SCC is modified on
  Lines~\ref{li:tdscc-newvar}--\ref{li:tdscc-modcfg}, and \algtd{}
  on \pref{li:tdscc-rec} returns \raisebox{-5.5pt}{
    \begin{tikzpicture}[auto,node distance=.7cm,font=\small]
      \tikzstyle{every node} = [circle, draw, inner sep=0pt,minimum size=2.5ex]
      \node (6) {$\small 6$};
      \node [right of=6] (7) {$\small 7$};
      \node [right of=7] (6') {$\small x_6$};
      \node [right of=6'] (8) {$\small 8$};
      \node [right of=8] (5') {$\small x_5$};
      \path (6) edge[forward] (7); 
      \path (7) edge[forward] (6'); 
      \path (6') edge[forward] (8); 
      \path (8) edge[forward] (5'); 
      \path (6') edge[backward, bend left=29] (6); 
    \end{tikzpicture}
  } for the modified graph.
  Moving $x_5$ from $V'$ to $X'$ on \pref{li:tdscc-move},
  adding $5 \forwardarrow 6$ on \pref{li:tdscc-forwardarrow},
  and adding $x_5 \backwardarrow 5$ on \pref{li:tdscc-backwardarrow}
  yields the WPO for the SCC.
  \qef
\end{example}

\begin{figure}[t]
  \centering
  \begin{subfigure}{.33\textwidth}
    \centering
    \begin{tikzpicture}[auto,node distance=.8cm,font=\small]
      \tikzstyle{every node} = [rectangle, draw, inner sep=0pt,minimum size=2.5ex]
      \node (1) at (0,0) {$1$};
      \node [right of=1] (2) {$2$};
      \node [below of=2, node distance=0.6cm] (5) {$5$};
      \node [right of=5] (6) {$6$};
      \node [right of=6] (7) {$7$};
      \node [right of=7] (8) {$8$};
      \node [right of=2, node distance=2cm] (3) {$3$};
      \node [above right of=8] (4) {$4$};
      \path (1) edge[cfgedge] (2); 
      \path (2) edge[cfgedge] (3); 
      \path (3) edge[cfgedge, bend right=25] (2); 
      \path (3) edge[cfgedge] (4); 
      \path (1) edge[cfgedge] (5); 
      \path (5) edge[cfgedge] (6); 
      \path (6) edge[cfgedge] (7); 
      \path (7) edge[cfgedge] (3); 
      \path (7) edge[cfgedge] (8); 
      \path (8) edge[cfgedge] (4); 
      \path (7) edge[cfgedge, bend right=33] (6); 
      \path (8) edge[cfgedge, bend left=25] (5); 
      \path (4) edge[cfgedge, loop above] (4); 
    \end{tikzpicture}
    \caption{}
    \label{fig:cfg3}
  \end{subfigure}
    \begin{subfigure}{.66\textwidth}
    \centering
      \begin{tikzpicture}[auto,node distance=.8cm,font=\small]
        \tikzstyle{every node} = [circle, draw, inner sep=0pt,minimum size=2.5ex]
        \node (1) at (0,0) {$1$};
        \node [right of=1] (2) {$2$};
        \node [below of=2, node distance=0.6cm] (5) {$5$};
        \node [right of=5] (6) {$6$};
        \node [right of=6] (7) {$7$};
        \node [right of=7] (6') {$x_6$};
        \node [right of=6'] (8) {$8$};
        \node [right of=8] (5') {$x_5$};
    
        \node (3) at (5,0) {$3$};
        \node [right of=3] (2') {$x_2$};
        \node [right of=2'] (4) {$4$};
        \node [right of=4] (4') {$x_4$};
        
        \path (1) edge[forward] (2); 
        \path (2) edge[forward] (3); 
        \path (3) edge[forward] (2'); 
        \path (2') edge[forward] (4); 
        \path (1) edge[forward] (5); 
        \path (5) edge[forward] (6); 
        \path (6) edge[forward] (7); 
        \path (7) edge[forward] (6'); 
        \path (6') edge[forward] (8); 
        \path (8) edge[forward] (5'); 
        \path (5') edge[forward] (3); 
        \path (5') edge[forward] (4); 
        \path (4) edge[forward] (4'); 
        \path (2') edge[backward, bend right=20] (2); 
        \path (6') edge[backward, bend right=30] (6); 
        \path (5') edge[backward, bend right=20] (5); 
        \path (4') edge[backward, bend right=30] (4); 
      \end{tikzpicture}
      \caption{}
      \label{fig:wpo3}
    \end{subfigure}
  \vspace{-2ex}
  \caption{\subref{fig:cfg3}~Directed graph $G_3$ and \subref{fig:wpo3}~WPO 
  $\wpo_3$. Vertices $V$
  are labeled using DFN;
   exits $X = \{x_2, x_4, x_5, x_8 \}$.}
  \label{fig:cfg-wpo3}
\end{figure}

Before we prove that the output of \algtd{} in \pref{alg:td} is a WPO,
we prove that the output of function \CWPOSCCKw\ satisfies the property of a WPO
for non-trivial strongly connected graph in \pref{cor:nocycle}.

\begin{lemma}
  \label{lem:structure}
  Given a non-trivial strongly connected graph $G(V, \cfgarrow)$ and its
  depth-first forest $D$, 
  the returned value $(V, X, \forwardarrow, \backwardarrow)$, $\etch$, $x_\etch$ of
  \CWPOSCCKw{$G, D$} satisfies
  $x_\etch \backwardarrow \etch$ and $\component{\etch}{x_\etch}_{\forwardarrowstar} = V \cup X$.
\end{lemma}
\begin{proof} 
  We use structural induction on the input $G$ to prove this.
  The base case is when $G$ has no non-trivial nested SCCs,
  and the inductive step uses an induction hypothesis on the non-trivial nested SCCs.

\noindent \textsc{[Base case]:}
  If the graph only has a single vertex $\etch$ and a self-loop, then
  \SWPOKw{$\etch$} is returned on \pref{li:tdscc-slwpo},
  whose value satisfies the lemma.
  Otherwise, because there are no non-trivial nested SCCs inside the graph,
  removing the head $\etch$ on \pref{li:tdscc-removal}
  removes all the cycles in the graph.
  Also, adding a new vertex $x_\etch$ on \pref{li:tdscc-newvar} and \ref{li:tdscc-modcfg}
  does not create a cycle, so the modified graph is acyclic.
  Therefore, $(V', \forwardarrow')$ on \pref{li:tdscc-rec} equals $(V', \cfgarrow')$.
  Because the input graph is strongly connected,
  every vertex in the graph is reachable from $\etch$.
  This is true even without the back edges because we can ignore the cyclic paths.
  Also, because the exit $x_\etch$ is a unique sink in the modified graph,
  $x_\etch$ is reachable from every vertex.
  It is unique because the negation would imply that there is a vertex in the
  original graph that has no outgoing edges, contradicting that the input graph is strongly
  connected.
  Finally, with changes in \pref{li:tdscc-forwardarrow} and \ref{li:tdscc-backwardarrow},
  we see that the lemma holds for the base case.

\noindent \textsc{[Inductive step]:}
  Let $G_i$ be one of the maximal non-trivial nested SCCs that is identified
  on \pref{li:td-scc}. 
  Because $h_i$ has the minimum DFN in $G_i$,
  $h_i$ must be an entry of $G_i$,
  and there must exist an incoming scheduling constraints to $h_i$.
  By the induction hypothesis, $u \forwardarrowplus h_i$ implies $u \forwardarrowplus v$
  for all $v \in (V_i \cup X_i)$ and $u \notin (V_i \cup X_i)$.
  Also, because scheduling constraints added on \pref{li:td-forwardarrow} has exits
  as their sources,
  $v \forwardarrowplus w$ implies $x_i \forwardarrowplus w$
  for all $v \in (V_i \cup X_i)$ and $w \notin (V_i \cup X_i)$.
  The graph of super-nodes (an SCC contracted to a single vertex)
  of the modified graph $(V', \cfgarrow')$ is acyclic.
  Therefore, by applying the similar reasoning as the base case on this graph of super-nodes,
  we see that the lemma holds.
\end{proof}

Armed with the above lemma, we now prove that \algtd{} constructs a WPO.

\begin{theorem}
  Given a graph $G(V, \cfgarrow)$ and its depth-first forest $D$,
  the returned value $(V, X, \forwardarrow, \backwardarrow)$
  of $\algtd{}(G, D)$ is a WPO for $G$.
\end{theorem}
\begin{proof}  We show that the returned value $(V, X, \forwardarrow, \backwardarrow)$
  satisfies all properties \pref{it:Wx}--\pref{it:Wbd} in \pref{def:WPO}.
  
  [\ref{it:Wx}]~$V$ equals the vertex set of the input graph, and $X$ consists
  only of the newly created exits.

  [\ref{it:Wred}]~For all exits, $x_\etch \backwardarrow \etch$ is added
  on \pref{li:swpo-ret} and \ref{li:tdscc-backwardarrow}. These are the only places
  \backward{}s are created.

  [\ref{it:Wblue}]~All \forward{}s are created on \pref{li:td-forwardarrow}, \ref{li:swpo-ret},
  and \ref{li:tdscc-forwardarrow}.

  [\ref{it:Whpo}-\ref{it:Hpartial}]~
  $(\forwardarrow_i^*, V_i \cup X_i)$ is reflexive and transitive by definition.
  Because the graph with maximal SCCs contracted to single vertices (super-nodes)
  is acyclic, \forward{}s on \pref{li:td-forwardarrow} cannot create a cycle.
  Also, \pref{li:tdscc-forwardarrow} only adds outgoing \forward{}s and does not
  create a cycle.
  Therefore, $(\forwardarrowstar, V \cup X)$ is antisymmetric.

  [\ref{it:Whpo}-\ref{it:Hoto}]~Exactly one \backward{} is created per
  exit on \pref{li:swpo-ret} and \ref{li:tdscc-backwardarrow}.
  Because $\etch$ is removed from the graph afterwards, it does not become a target of
  another \backward{}.

  [\ref{it:Whpo}-\ref{it:Hdir}]~By \pref{lem:structure}, $x_\etch \backwardarrow \etch$ implies
  $\etch \forwardarrowplus x_\etch$.

  [\ref{it:Whpo}-\ref{it:Hnest}]~
  Because the maximal SCCs on \pref{li:td-rec} are disjoint,
  by \pref{lem:structure}, all components $\component{h_i}{x_i}_{\forwardarrowstar}$ are disjoint.

  [\ref{it:Whpo}-\ref{it:Hexit}]~
  All additional \forward{}s going outside of a component
  have exits as their sources on \pref{li:td-forwardarrow}.

  [\ref{it:Wbd}]~For $u \cfgarrow v$, either (i)~\forward{} is added in
  \pref{li:td-forwardarrow} and \ref{li:tdscc-forwardarrow}, or (ii)~\backward{} is added in
  \pref{li:tdscc-slwpo} and \ref{li:tdscc-backwardarrow}.
  In the case of (i),
  one can check that the property holds for \pref{li:tdscc-forwardarrow}.
  For \pref{li:td-forwardarrow}, if $u$ is a trivially maximal SCC, $x_i = u$.
  Else, $u \forwardarrowstar x_i$ by \pref{lem:structure}, and
  with added $x_i \forwardarrow v$, $u \forwardarrowplus v$.
  In the case of (ii), $u \in \component{\etch}{x_\etch}_{\forwardarrowstar}$ by
  \pref{lem:structure} where $v = \etch$.
\end{proof}

The next theorem proves that the WPO constructed by \algtd{} does not
include superfluous \forward{}s, which could reduce parallelism during the
fixpoint computation.

\begin{theorem}
  \label{thm:partial2}
  For a graph $G(V, \cfgarrow)$ and its depth-first forest $D$, WPO $\wpo(V, X, \forwardarrow, \backwardarrow)$
  returned by $\algtd{}(G, D)$ has the smallest $\forwardarrowstar$
  among the WPOs for $G$ with the same set of $\backwardarrow$.
\end{theorem}
\begin{proof} We use structural induction on the input $G$ to prove this.

\noindent \textsc{[Base case]:}
  The base case is when $G$ is either (i)~a trivial SCC or (ii)~a single vertex with a self-loop.
  If $G$ is a trivial SCC, then the algorithm outputs single trivial WPO, whose $\forwardarrowstar$ is empty.

  If $G$ is a single vertex with a self-loop, then there should be at least one exit in the WPO.
  The algorithm outputs self-loop WPO, whose $\forwardarrowstar$ contains only $(\etch, x_\etch)$.

\noindent \textsc{[Inductive step]:}
  The two cases for the inductive step are (i)~$G$ is strongly connected and
  (ii)~$G$ is not strongly connected.
  If $G$ is strongly connected, then due to \pref{cor:nocycle},
  all WPOs of $G$ must have a single maximal component that is equal to $V \cup X$.
  We only consider WPOs with the same set of $\backwardarrow$, 
  so $\etch$, with minimum DFN, is the head of all WPOs that we consider.
  By the inductive hypothesis, the theorem holds for
  the returned value of recursive call on \pref{li:tdscc-rec},
  where the input is the graph without $\etch$.
  Because all WPOs have to satisfy $v \forwardarrowstar x_\etch$ for all $v \in V \cup X$,
  adding $x_\etch$ does not affect the size of $\forwardarrowstar$.
  Also, \pref{li:tdscc-forwardarrow} only adds the required \forward{}s to satisfy
  \pref{it:Wbd} for the dependencies whose source is $\etch$.

  If $G$ is not strongly connected, then by the induction hypothesis,
  the theorem holds for the returned values of \SWPOKw\ on \pref{li:td-rec}
  for all maximal SCCs.
  \pref{li:td-forwardarrow} only adds the required \forward{}s to satisfy
  \pref{it:Wbd} and \pref{it:Hexit}
  for dependencies between different maximal SCCs,
\end{proof}

\subsection{Bottom-up Iterative Construction}
\label{sec:AlgorithmBU}

\SetKwProg{Def}{def}{:}{end}
\SetKwFunction{BKw}{B}
\SetKwFunction{CFKw}{CF}
\SetKwFunction{RKw}{R}
\SetKwFunction{OKw}{O}
\SetKwFunction{XKw}{exit}
\SetKwFunction{TKw}{T}
\SetKwFunction{repKw}{rep}
\SetKwFunction{mergeKw}{merge}
\SetKw{inKw}{in}
\SetKw{suchKw}{such}
\SetKw{thatKw}{that}
\SetKw{thereKw}{there}
\SetKw{stKw}{s.t.}
\SetKw{existsKw}{exists}
\SetKw{andKw}{and}
\SetKw{continueKw}{continue}
\SetKw{returnKw}{return}
\SetKw{questionKw}{?}
\SetKw{colonKw}{:}
\SetKwData{liftKw}{lift}
\SetKwFunction{RemoveAllCrossForwardEdgesKw}{removeAllCrossForwardEdges}
\SetKwFunction{ConstructWPOForSCCKw}{constructWPOForSCC}
\SetKwFunction{RestoreCrossForwardEdgesKw}{restoreCrossForwardEdges}
\SetKwFunction{ConnectWPOsOfMaximalSCCsKw}{connectWPOsOfMaximalSCCs}
\SetKwFunction{FindNestedSCCsKw}{findNestedSCCs}

\pref{alg:bu} presents $\algbu{}(G, D, \liftKw)$, an efficient, almost-linear time
iterative algorithm for constructing the WPO of a graph $G(V, \cfgarrow)$; the
role of the boolean parameter $\liftKw$ is explained in \pref{sec:wto}, and is
assumed to be to be \textit{false} throughout this section. 
\algbu{} is also parametrized by the depth-first forest $D$ of the graph $G$,
and the WPO constructed by the algorithm may change with different forest.
DFF defines back edges ($\BKw$) and cross or forward edges ($\CFKw$)
on Lines~\ref{li:bu-B} and \ref{li:bu-CF}.

$\algbu$ maintains a partitioning of the vertices with each subset containing
vertices that are currently known to be strongly connected. The algorithm use a
disjoint-set data structure to maintain this partitioning. The operation
$\repKw(v)$ returns the representative of the subset that contains vertex $v$,
which is used to determine whether or not two vertices are in the same
partition. The algorithm assumes that the vertex with minimum DFN is the
representative of the subset. Initially, $\repKw(v) = v$ for all vertices $v \in
V$. When the algorithm determines that two subsets are strongly connected, they
are merged into a single subset. The operation $\mergeKw(v,h)$ merges the
subsets containing $v$ and $h$, and assigns $h$ to be the representative for the
combined subset.

\begin{figure}[t]
  \begin{minipage}[t]{0.96\textwidth}
\SetInd{.3em}{0.5em}
  \vspace{0pt}
  \begin{algorithm}[H]
      \small
    \DontPrintSemicolon
    \setlength{\columnsep}{15pt}
    \Indmm
    \KwIn{Directed graph $G(V, \cfgarrow)$, Depth-first forest $D$, Boolean \liftKw}
    \KwOut{WPO $\wpo(V, X, \forwardarrow, \backwardarrow)$}
    \Indpp
    \begin{multicols}{2}
      \small
    \tikzmk{A}  
    $X, \forwardarrow, \backwardarrow \coloneqq \emptyset, \emptyset, \emptyset$\label{li:bu-is}\;%
    \Comment*[l]{These are also feedback edges of $G$.}
    $\BKw \coloneqq \{ (u,v) \mid u \cfgarrow v \textrm{ is a back edge in $D$} \}$\label{li:bu-B}\;
    $\CFKw \coloneqq \{ (u,v) \mid u \cfgarrow v \textrm{ is a cross/fwd edge in $D$} \}$\label{li:bu-CF}\;
    \ForEach{$v \in V$}{
      $\repKw(v) \coloneqq \XKw[v] \coloneqq v$; $\RKw[v] \coloneqq \emptyset$\label{li:bu-ie}
    }
    \RemoveAllCrossForwardEdgesKw{}\label{li:bu-rm}\;
    {%
    \SetNlSty{textbf}{$\star$}{}
    \ForEach{$v \in V$}{
      $\OKw[v] \coloneqq \{(u, v) \in (\cfgarrow \setminus B)\}$\label{li:bu-io}
    }
    }%
      \ForEach{$\etch \in V$ \inKw \text{descending} $DFN(D)$}{\label{li:bu-lm}
      \RestoreCrossForwardEdgesKw{$\etch$}\label{li:bu-rs}\;
      \ConstructWPOForSCCKw{$\etch$}\label{li:bu-wpo}
    }
    {%
    \SetNlSty{textbf}{$\star$}{}
    \ConnectWPOsOfMaximalSCCsKw{}\label{li:bu-c}\;
    \returnKw{$(V, X, \forwardarrow, \backwardarrow)$}
    }%
    \BlankLine
    \tikzmk{B} \boxit{gray}{.94}
    \tikzmk{A}
    \Def{\ConstructWPOForSCCKw{$\etch$}}{
      \tikzmk{B} \boxit{cyan}{.89}
      $N_\etch, P_\etch \coloneqq$ \FindNestedSCCsKw{$\etch$}\label{li:bu-scc}\;
      \lIf{$P_\etch = \emptyset$}{
        \returnKw \label{li:bu-ret} %
      } 
      {%
      \SetNlSty{textbf}{$\star$}{}
      \ForEach{$v \in N_\etch$}{\label{li:bu-lb}
        \Comment*[l]{$(u,v')$ is the original edge. }
        \Comment*[l]{$\repKw(u)$ represents maximal SCC}
        \Comment*[l]{containing $u$ but not $v'$.}
        $\forwardarrow \coloneqq \forwardarrow \cup \{(\XKw[\repKw(u)], v') \mid (u, v') \in \OKw[v]\}$\label{li:bu-b}\;
        \SetInd{.1em}{0.22em}
        \If{\liftKw}{
          $\forwardarrow \coloneqq\! \forwardarrow \cup \{(\XKw[\repKw(u)], v)\! \mid\! (u, v) \in (\cfgarrow \setminus B)\}$
        }
        \SetInd{.3em}{0.5em}
      }
      $x_\etch \coloneqq$ new exit; $X \coloneqq X \cup \{x_\etch\}$\label{li:bu-newx}\;
      $\forwardarrow \coloneqq \forwardarrow \cup \{(\XKw[\repKw(p)], x_\etch)\mid p \in P_\etch\}$\label{li:bu-be}\;
      $\backwardarrow \coloneqq \backwardarrow \cup \{(x_\etch, \etch)\}$\label{li:bu-r}\; 
      }%
      {%
      \SetNlSty{textbf}{$\star$}{}
      $\XKw[\etch] \coloneqq x_\etch$\label{li:bu-exit}\;
      }%
      \lForEach{$v \in N_\etch$}{
        \mergeKw{$v, \etch$}\label{li:bu-merge}
      }
    }
    \tikzmk{A}
    \Def{\RemoveAllCrossForwardEdgesKw{}}{
    \tikzmk{B} \boxit{cyan}{.89}
      \ForEach{$(u, v) \in \CFKw$}{
        $\cfgarrow \coloneqq \cfgarrow \setminus \{ (u,v) \}$\label{li:bu-rmcf}\;
        \Comment*[l]{Removed edges will be restored }
        \Comment*[l]{when $\etch$ is the LCA of $u, v$ in $D$.}
        $\RKw[lca_{T}(u,v)] \coloneqq \RKw[lca_{T}(u,v)] \cup \{(u,v)\}$\label{li:bu-lca}\;
      }
    }
    \tikzmk{A}
    \Def{\RestoreCrossForwardEdgesKw{$\etch$}}{
      \tikzmk{B} \boxit{cyan}{.89}
      \ForEach{$(u, v) \in \RKw[\etch]$}{
        $\cfgarrow \coloneqq \cfgarrow \cup \{(u, \repKw(v))\}$\label{li:bu-rscf}\;
        {%
        \SetNlSty{textbf}{$\star$}{}
        \Comment*[l]{Record the original edge.}
        $\OKw[\repKw(v)] \coloneqq \OKw[\repKw(v)] \cup \{(u, v)\}$\label{li:bu-o}\;
        }%
      }
    }
    \tikzmk{A}
    \Def{\FindNestedSCCsKw{$\etch$}}{
    \tikzmk{B} \boxit{cyan}{.89}
    {%
    \SetNlSty{textbf}{$\star$}{}
    \Comment*[l]{Search backwards from the sinks.}
    $P_\etch \coloneqq \{ \repKw(p) \mid (p, \etch) \in \BKw\}$\label{li:bu-bp}\;
    }%
    $N_\etch \coloneqq \emptyset$\Comment*[r]{Reps.\ of nested SCCs except $\etch$.}
    $W_\etch \coloneqq P_\etch \setminus \{ \etch \}$\Comment*[r]{Worklist.}
    \While{$\thereKw\ \existsKw\ v \in W_\etch$}{
      $W_\etch,\ N_\etch \coloneqq W_\etch \setminus \{v\},\ N_\etch \cup \{v\}$\label{li:bu-pop}\;
      \ForEach{$p$ \stKw $(p, v) \in (\cfgarrow \setminus \BKw)$}{
        \If{$\repKw(p) \notin N_\etch \cup \{\etch \} \cup W_\etch$}{
          $W_\etch \coloneqq W_\etch \cup \{\repKw(p)\}$\label{li:bu-found}
        }
      }
    }
    \returnKw $N_\etch, P_\etch$
    }
    {%
    \SetNlSty{textbf}{$\star$}{}    
    \tikzmk{A}
    \Def{\ConnectWPOsOfMaximalSCCsKw{}}{
    \tikzmk{B} \boxit{cyan}{.89}
    \ForEach{$v \in V$ \stKw $\repKw(v) = v$}{\label{li:bu-lbb}
      $\forwardarrow \coloneqq \forwardarrow \cup \{(\XKw[\repKw(u)], v') \mid (u, v') \in \OKw[v]\}$\label{li:bu-bb}\;%
      \SetInd{.1em}{0.22em}
      \If{\liftKw}{
        $\forwardarrow \coloneqq\! \forwardarrow \cup \{(\XKw[\repKw(u)], v)\! \mid\! (u, v) \in (\cfgarrow \setminus B)\}$
      }
      \SetInd{.3em}{0.5em}
    }
    }%
    }%
    \end{multicols}
    \vspace{-2ex}
    \caption{\algbu{}$(G, D, \texttt{lift})$}
    \label{alg:bu}
  \end{algorithm}
\SetInd{.5em}{1em}  
  \end{minipage}
  \vspace{-1ex}
\end{figure}

\begin{example}
  Consider the graph $G_3$ in \pref{fig:cfg3}. Let $\underline{1} \mid
  \underline{2} \mid \underline{3} \mid \underline{4} \mid \underline{5} \mid
  \underline{6}\ 7 \mid \underline{8}$ be the current partition, where 
  the underline marks the representatives. Thus,
  $\repKw(1) = 1$ and $\repKw(6) = \repKw(7) = 6$. If it is found that vertices $5,
  6, 7, 8$ are strongly connected, then calls to $\mergeKw(6, 5)$ and $\mergeKw(8, 5)$
  update the partition to $\underline{1} \mid \underline{2} \mid
  \underline{3} \mid \underline{4} \mid \underline{5}\ 6\ 7\ 8$. Thus,
  $\repKw(6) = \repKw(7) = \repKw(8) = \repKw(5) = 5$. \qef
\end{example}

Auxiliary data structures $\repKw$, $\XKw$, and $\RKw$ are initialized on
\pref{li:bu-ie}. The map $\XKw$ maps an SCC (represented by its header $\etch$)
to its corresponding exit $x_\etch$. Initially, $\XKw[v]$ is set to $v$, and
updated on \pref{li:bu-exit} when a non-trivial SCC is discovered by the
algorithm. 

The map $\RKw$ maps a vertex to a set of edges, and is used to handle
irreducible graphs \cite{hecht1972flow,tarjan1973testing}. Initially, $\RKw$ is
set to $\emptyset$, and updated on \pref{li:bu-lca}. The function
$\FindNestedSCCsKw$ relies on the assumption that the graph is reducible. This
function follows the edges backwards to find nested SCCs using $\repKw(p)$
instead of predecessor $p$, as on Lines~\ref{li:bu-bp} and \ref{li:bu-found}, to
avoid repeatedly searching inside the nested SCCs. $\repKw(p)$ is the unique
entry to the nested SCC that contains the predecessor $p$ if the graph is
reducible. To make this algorithm work for irreducible graphs as well, cross or
forward edges are removed from the graph initially by function
$\RemoveAllCrossForwardEdgesKw$ (called on \pref{li:bu-rm}) to make the graph
reducible. Removed edges are then restored by function
$\RestoreCrossForwardEdgesKw$ (called on \pref{li:bu-rs}) right before the edges
are used. The graph is guaranteed to be reducible when restoring a removed edge
$u \cfgarrow v$ as $u \cfgarrow \repKw(v)$ when $\etch$ is the lowest common
ancestor (LCA)  of $u, v$ in the depth-first forest. Cross and forward edges are removed on
\pref{li:bu-rmcf} and are stored at their LCAs in $D$ on \pref{li:bu-lca}.
Then, as $\etch$ hits the LCAs in the loop, the removed edges are restored on
\pref{li:bu-rscf}. Because the graph edges are modified, map $\OKw$ is used to
track the original edges. $\OKw[v]$ returns set of original non-back edges that
now targets $v$ after the modification. The map $\OKw$ is initialized on
\pref{li:bu-io} and updated on \pref{li:bu-o}.

The call to function $\ConstructWPOForSCCKw(\etch)$ on \pref{li:bu-wpo}
constructs a WPO for the largest SCC such that $\etch = \argmin_{v\in V'}
DFN(D, v)$, where $V'$ is the vertex set of the SCC. For example,
$\ConstructWPOForSCCKw(5)$ constructs WPO for SCC with vertex set $\{5, 6, 7,
8\}$. Because the loop on \pref{li:bu-lm} traverses the vertices in descending
DFN, the WPO for a nested SCC is constructed before that of the enclosing SCC.
For example, $\ConstructWPOForSCCKw(6)$ and $\ConstructWPOForSCCKw(8)$ are
called before $\ConstructWPOForSCCKw(5)$, which construct WPOs for nested SCCs
with vertex sets $\{6, 7\}$ and $\{8\}$, respectively. Therefore,
$\ConstructWPOForSCCKw(\etch)$ reuses the constructed sub-WPOs.

The call to function $\FindNestedSCCsKw(\etch)$ on \pref{li:bu-scc} 
returns $N_\etch$, the representatives of the nested SCCs $N_\etch$, as well as 
$P_\etch$, the predecessors of $\etch$ along back edges.
If $P_\etch$ is
empty, then the SCC is trivial, and the function immediately returns on
\pref{li:bu-ret}. \pref{li:bu-b} adds \forward{}s for the dependencies crossing
the nested SCCs. As in \algtd{}, this must be from the exit of maximal SCC that
contains $u$ but not $v'$ for $u \cfgarrow v'$. Because $u \cfgarrow v'$ is now
$u \cfgarrow \repKw(v')$, $\OKw[v]$, where $v = \repKw(v')$, is looked up to
find $u \cfgarrow v'$. $\XKw$ is used to find the exit, where $\repKw(u)$ is the
representative of maximal SCC that contains $u$ but not $v$.
If the parameter \liftKw is \textit{true}, then \forward{} targeting $v$ is also
added, forcing all scheduling predecessors outside of a component to be visited
before the component's head.
Similarly, function
$\ConnectWPOsOfMaximalSCCsKw$ is called after the loop on \pref{li:bu-bb} to
connect WPOs of maximal SCCs. The exit, \forward{}s to the exit, and
\backward{}s are added on Lines~\ref{li:bu-newx}--\ref{li:bu-r}. After the
WPO is constructed, \pref{li:bu-exit} updates the map $\XKw$, and
\pref{li:bu-merge} updates the partition.

\begin{table}
  \caption{Steps of \algbu{} for graph $G-3$ in \pref{fig:cfg-wpo3}.}
  \label{tab:effex}
  \let\center\empty
  \let\endcenter\relax
  \centering
  \input{efficient_ex}
  \vspace{-1ex}
\end{table}

\begin{example}
  \pref{tab:effex} describes the steps of \algbu{} for the irreducible graph
  $G_3$ (\pref{fig:cfg-wpo3}). The~$\cfgarrow$~updates column shows the
  modifications to the graph edges, the Current partition column shows the
  changes in the disjoint-set data structure, and the Current WPO column shows
  the WPO constructed so far.
  Each row shows the changes made in each step of the algorithm.
  Row `Init' shows the initialization step on Lines~\ref{li:bu-ie}--\ref{li:bu-rm}.
  Row `$\etch=k$' shows the $k$-th iteration of the loop on~Lines~\ref{li:bu-lm}--\ref{li:bu-wpo}.
  The loop iterates over the vertices in descending DFN: $8, 7, 6, 5, 4, 3, 2, 1$.
  Row `Final' shows the final step after the loop on \pref{li:bu-c}.

  During initialization, the cross or forward edges $\{(7,3), (8,4)\}$ are
  removed, making $G_3$ reducible. These edges are added back as
  $\{(7,\repKw(3)), (8,\repKw(4))\}$ = $\{(7,2), (8,4)\}$ in $\etch=1$, where
  $1$ is the LCA of both $(7,3)$ and $(8,4)$. 
  $G_3$ remains reducible after restoration.
  In step $\etch=5$, WPOs for the nested SCCs are connected with
  $5 \forwardarrow 6$ and $x_6 \forwardarrow 8$.
  The new exit $x_5$ is created, connected to the WPO via $8 \forwardarrow x_5$, 
  and $x_5 \backwardarrow 5$ is added.
  Finally, $1 \forwardarrow 2$,
  $1 \forwardarrow 5$, $x_5 \forwardarrow 3$,
  $x_5 \forwardarrow 4$, and $x_2 \forwardarrow 4$
  are added, connecting the WPOs for maximal SCCs.
  If \liftKw is \textit{true}, then \forward{} $x_5 \forwardarrow 2$
  is added.
  \qef
\end{example}

\algbu{} adapts Tarjan-Havlak-Ramaligam's almost-linear time algorithm \cite{ramalingamTOPLAS1999}
for constructing Havlak's loop nesting forest (LNF)~\citep{havlak1997nesting}.
Similar to components in the WPO, an LNF represents the nesting relationship among the SCCs.
A WPO contains additional information in the form of scheduling 
constraints, which
are used to generate the concurrent iteration strategy for fixpoint
computation. Lines in \algbu{} indicated by $\star$ were added to the algorithm
in \citet{ramalingamTOPLAS1999}.
The following two theorems prove the correctness and runtime efficiency of \algbu{}.

\begin{theorem}
  \label{thm:same}
  Given a graph $G(V, \cfgarrow)$ and its depth-first forest $D$, \algtd{}$(G, D)$ and \algbu{}$(G, D, \textit{false})$ construct the same WPO for $G$.
\end{theorem}
\begin{proof}
  Using the constructive characterization of an LNF \cite[Definition 3]{ramalingamTOPLAS2002},
  an LNF of a graph is constructed by identifying the maximal SCCs,
  choosing headers for the maximal SCCs,
  removing incoming back edges to the headers to break the maximal SCCs,
  and repeating this on the subgraphs.
  In particular, Havlak's LNF is obtained when the vertex with minimum DFN 
  in the SCC is chosen as the header \cite[Definition 6]{ramalingamTOPLAS2002}.
  This construction is similar to how the maximal SCCs are identified on \pref{li:td-scc} of \algtd{}
  and how \algtd{} is called recursively on \pref{li:tdscc-rec}.
  Because \algbu{} is based on the LNF construction algorithm,
  both algorithms identify the same set of SCCs,
  resulting in the same $X$ and $\backwardarrow$.

  Now consider $u \cfgarrow v'$ in the original graph that is not a back edge.
  Let $v$ be $\repKw(v')$ when restoring the edge if the edge is cross/forward,
  or $v'$ otherwise.
  If both $u, v'$ are in some non-trivial SCC,
  then $v \in N_\etch$ on \pref{li:bu-lb} for some $\etch$.
  In this case, $\repKw(u)$ must be in $N_\etch$.
  If not, $u \cfgarrow v$ creates an entry to the SCC other than $\etch$.
  We know that $\etch$ must be the entry because it has minimum DFN.
  This contradicts that the modified graph is reducible over the whole run \cite[Section 5]{ramalingamTOPLAS1999}.
  Because all nested SCCs have been already identified,
  $\repKw(u)$ is the representative of the maximal SCC that contains $u$ but not $v$.
  Also, if no SCC contains both $u, v'$, then $v = \repKw(v)$ on \pref{li:bu-lbb}.
  Because all maximal SCCs are found, $\repKw(u)$ returns the representative of
  the maximal SCC that contains $u$ but not $v$.
  Therefore, \pref{li:bu-b} and \pref{li:bu-bb} construct the same $\forwardarrow$ as \algtd{}.
\end{proof}

\begin{theorem}
  \label{thm:AlmostLinearBottomUpWPO}
  The running time of \algbu{}$(G, D, *)$ is almost-linear. 
\end{theorem}
\begin{proof}
  The non-starred lines in \algbu{} are the same as Tarjan-Havlak-Ramalingam's
  almost-linear time algorithm for constructing Havlak's LNF~\cite[Section 5]{ramalingamTOPLAS1999}.
  The starred lines only add constant factors to the algorithm.
  Thus, the running time of \algbu{} is also almost-linear.
\end{proof}

\section{Connection to Weak Topological Order}
\label{sec:wto}

The weak partial order defined in \pref{sec:wpo} is a strict generalization of weak
topological order~(WTO) defined by \citet{bourdoncle1993efficient}. 
Let us first recall the definitions by Bourdoncle.

\begin{definition}
  \label{def:BourdoncleHTO}
  A \emph{hierarchical total order (HTO)} of a set $S$ is a well-parenthesized
  permutation of this set without two consecutive ``(''. \qef
\end{definition}

An HTO is a string over the alphabet $S$
augmented with left and right parenthesis. An HTO of $S$
induces a total order $\preceq$ over the elements of $S$.
The elements between two matching parentheses are called a \emph{component}, and the first
element of a component is called the \emph{head}.
The set of heads of the components containing the element $l$ is
denoted by $\omega(l)$. 

\begin{definition}
  \label{def:BourdoncleWTO}
  A \emph{weak topological order (WTO)} of a directed graph is a
  HTO of its vertices such that for every edge $u \rightarrow
  v$, either $u \prec v$ or $v \preceq u$ and $v \in \omega(u)$. \qef
\end{definition}

A WTO factors out a feedback edge set from the graph using the matching
parentheses and topologically sorts the rest of the graph to obtain a total
order of vertices. A feedback edge set defined by a WTO is $\{(u, v) \in
\cfgarrow \mid v \preceq u \text{ and } v \in \omega(u)\}$. 

\begin{example}
  \label{ex:BourdoncleWTO}
  The WTO of graph $G_1$ in \pref{fig:cfg1} is 
$1\ (2\ (3\ 4)\ (6\ 7\ 9\ 8)\ 5)\ 10$.
The feedback edge set defined by this WTO is $\{(4,3), (8,6), (5,2)\}$, 
which is the same as that defined by WPO $\wpo_1$.
\qef
\end{example}

\newcommand{\algwto}{\texttt{ConstructWTO\textsuperscript{TD}}}

\newcommand{\ptt}{\texttt{ConstructWTO}\textsuperscript{BU}}
\newcommand{\ext}{\forwardarrow_{\texttt{ext}}}
\pref{alg:wto} presents a top-down recursive algorithm for
constructing a WTO for a graph $G$ and its depth-first forest $D$.
Notice the use of increasing post DFN order
when merging the results on \pref{li:wto-for}. In general, a
reverse post DFN order of a graph is its topological order. Therefore,
\algwto{}, in effect, topologically sorts the DAG of SCCs recursively. Because
it is recursive, it preserves the components and their nesting relationship.
Furthermore, by observing the correspondence between \algwto{} and \algtd{}, we
see that $\algwto{}(G, D)$ and $\algtd{}(G, D, *)$ construct the same components with
same heads and nesting relationship. 

The definition of HPO (\pref{def:HPO}) generalizes the definition of
HTO (\pref{def:BourdoncleHTO}) to partial orders, while the
definition of WPO (\pref{def:WPO}) generalizes the definition of WTO
(\pref{def:BourdoncleWTO}) to partial orders. In other words, the two
definitions define the same structure if we strengthen \pref{it:Hpartial} to a
total order. If we view the exits in $X$ as closing parenthesis ``)'' and $x
\backwardarrow h$ as matching parentheses $(h \dots)$, then the correspondence
between the two definitions becomes clear. The conditions that a HTO 
must be well-parenthesized and that it disallows two consecutive
``('' correspond to conditions \pref{it:Hnest}, \pref{it:Hoto}, and
\pref{it:Wx}. While \pref{it:Hexit} is not specified in Bourdoncle's definition,
it directly follows from the fact that $\preceq$ is a total order. Finally, the
condition in the definition of WTO matches \pref{it:Wbd}. 
Thus, using the notion of WPO, we can define a WTO as:

\SetKwProg{Def}{def}{:}{end}
\SetKwFunction{CWTOSCCKw}{sccWTO}

\newcommand{\boxitt}[2]{\tikz[remember picture,overlay]{\node[yshift=3pt,xshift=3pt,fill=#1,opacity=.15,fit={(A)($(B)+(#2\linewidth,-0.5\baselineskip)$)}] {};}\ignorespaces}

\begin{figure}[t]
  \begin{minipage}[t]{0.96\textwidth}
\SetInd{.3em}{0.5em}
  \begin{algorithm}[H]
    \small
    \SetKw{andKw}{and}
    \SetKw{returnKw}{return}
    \SetKw{stKw}{s.t.}
    \DontPrintSemicolon
    \setlength{\columnsep}{15pt}
    \Indmm
    \KwIn{Directed graph $G(V, \cfgarrow)$, Depth-first forest $D$}
    \KwOut{WTO}
    \Indpp
    \begin{multicols}{2}
      \small
      \tikzmk{A}
      $G_1, G_2, \ldots, G_k \coloneqq SCC(G)$\label{li:wto-scc}\Comment*[r]{Maximal SCCs.}
      \ForEach{$i \in [1, k]$} {
        \Comment*[l]{WTOs for SCCs.}
        $WTO_i, \etch\#_i \coloneqq$ \CWTOSCCKw{$G_i, D$}\label{li:wto-rec}
      }
      $WTO \coloneqq \text{nil}$\;
      \ForEach{$i$ in increasing order of $\etch\#_i$}{ \label{li:wto-for}
        \Comment*[l]{'$::$' means to append.}
        $WTO \coloneqq WTO_i::WTO$
      }
      \tikzmk{B} \boxitt{gray}{.97}
      \returnKw{$WTO$}\;
      \tikzmk{A}
      \Def{\CWTOSCCKw{$G, D$}}{
      \tikzmk{B} \boxit{cyan}{.90}
        \Comment*[l]{$G$ is strongly connected.}
        $\etch \coloneqq \argmin_{v\in V}DFN(D, v)$\label{li:wto-etch}\Comment*[r]{Minimum DFN.}
        $B \coloneqq \{ (v, \etch) \mid  v \in V \text{ and } v \cfgarrow \etch\}$\label{li:wto-B}\;
        \lIf{$|B| = 0$}{
          \returnKw{$\etch, \text{post-DFN}(D, \etch)$}\label{li:wto-twpo}
        }
        \lIf{$|V| = 1$}{
          \returnKw{$\big(\etch\big), \text{post-DFN}(D, \etch)$}\label{li:wto-slwpo}
        }
        $V' \coloneqq V \setminus \{\etch\}$\label{li:wto-removal}\;
        $\cfgarrow' \coloneqq \cfgarrow_{\downharpoonright V'}$\label{li:wto-modcfg}\;
        $WTO \coloneqq$ \algwto{$((V', \cfgarrow'), D_{\downharpoonright V'})$}\label{li:wto-rec2}\;
        \returnKw{$\big(\etch :: WTO\big), \text{post-DFN}(D, \etch)$}
      }
    \end{multicols}
    \vspace{-2ex}
    \caption{\algwto{}$(G, D)$}
    \label{alg:wto}
  \end{algorithm}
\SetInd{.5em}{1em}  
  \end{minipage}
  \vspace{-1ex}
\end{figure}

\begin{definition}
  \label{def:wto}
  A \emph{weak topological order} (WTO) for a graph $G(V, \cfgarrow)$ is a WPO $\wpo(V,
  X, \forwardarrow, \backwardarrow)$ for $G$ where $(V \cup X, \forwardarrowstar)$ is
  a total order. \qef
\end{definition}

\pref{def:wto} hints at how a WTO for a graph can be constructed from a WPO. The
key is to construct a linear extension of the partial order $(V \cup X,
\forwardarrowstar)$ of the WPO, while ensuring that properties
\pref{it:Hpartial}, \pref{it:Hnest}, and \pref{it:Hexit} continue to hold.
$\ptt{}$ (\pref{alg:ptt}) uses the above insight to construct a WTO of $G$ in almost-linear time, 
as proved by the following two theorems.

\begin{theorem}
  Given a directed graph $G$ and its depth-first forest $D$,
  the returned value $(V, X, \forwardarrow, \backwardarrow)$ of $\ptt{}(G, D)$
  is a WTO for $G$.
\end{theorem}
\begin{proof}
  The call $\algbu{}(G, D, \textit{true})$ on \pref{li:ptt-bu} constructs a WPO for $G$.
  With \liftKw set to \textit{true} in \algbu{}, all scheduling
  predecessors outside of a component are visited before the head of the
  component. The algorithm then visits the vertices in topological order
  according to~$\forwardarrow$. Thus, the
  additions to~$\forwardarrow$ on \pref{li:ptt-ext} do not violate
  \pref{it:Hpartial} and lead to a total order $(V \cup X, \forwardarrowstar)$.
  Furthermore, because a stack is used as the worklist and because of
  \pref{it:Hexit}, once a head of a component is visited, no element outside the
  component is visited until all elements in the component are visited.
  Therefore, the additions to~$\forwardarrow$ on \pref{li:ptt-ext} preserve the
  components and their nesting relationship, satisfying \pref{it:Hnest}. Because
  the exit is the last element visited in the component, no \forward{} is added
  from inside of the component to outside, satisfying \pref{it:Hexit}.
  Thus, $\ptt{}(G, D)$ constructs a WTO for $G$. 
  \end{proof}

\begin{example}
  For graph $G_1$ in \pref{fig:cfg1}, $\ptt{}(G_1, D)$ returns the following WTO:
  \vspace{-1ex}
  \begin{center}
  \begin{tikzpicture}[auto,node distance=.8cm,]
  \tikzstyle{every node} = [circle, draw, inner sep=0pt,minimum size=2.5ex]
  \node (1) {1};
  \node [right of=1] (2) {$2$};
  \node [right of=2] (3) {$3$};
  \node [right of=3] (4) {$4$};
  \node [right of=4] (3') {$x_{3}$};
  \node [right of=3'] (6) {$6$};
  \node [right of=6] (7) {$7$};
  \node [right of=7] (9) {$9$};
  \node [right of=9] (8) {$8$};
  \node [right of=8] (6') {$x_{6}$};
  \node [right of=6'] (5) {$5$};
  \node [right of=5] (2') {$x_{2}$};
  \node [right of=2'] (10) {$10$};
  \path (1) edge[forward] (2); 
  \path (2) edge[forward] (3); 
  \path (3) edge[forward] (4); 
  \path (4) edge[forward] (3'); 
  \path (3') edge[forward, bend left=20] (5); 

  \path (2) edge[forward, bend right=20] (6); 
  \path (6) edge[forward] (7);
  \path (6) edge[forward, bend right=23] (9);
  \path (7) edge[forward, bend left=23] (8);
  \path (9) edge[forward] (8);
  \path (8) edge[forward] (6');
  \path (6') edge[forward] (5);

  \path (5) edge[forward] (2');
  \path (2') edge[forward] (10);
  \path (6') edge[backward, bend left=30] (6);
  \path (3') edge[backward, bend right=25] (3);
  \path (2') edge[backward, bend right=18] (2);

  \path (3') edge[forward] (6);
  \path (7) edge[forward] (9);
  \end{tikzpicture}
  \end{center}
  \vspace{-1.5ex}
  \ptt{} first constructs the WPO $\wpo_1$ in \pref{fig:wpo1} using \algbu{}.
  The partial order $\forwardarrowstar$ of $\wpo_1$ is extended 
  to a total order by adding $x_3 \forwardarrow 6$ and $7 \forwardarrow 9$.
  The components and their nesting relationship in $\wpo_1$ are preserved in the 
  constructed WTO.
  This WTO is equivalent to $1~(2~(3~4)~(6~7~9~8)~5)~10$ in Bourdoncle's representation
  (see \pref{def:BourdoncleWTO} and \pref{ex:BourdoncleWTO}).
  \qef
\end{example}

\begin{theorem}
  \label{thm:AlmostLinearWTO}
  Running time of \ptt{}$(G, D)$ is almost-linear.
\end{theorem}
\begin{proof}
  The call to \algbu{} in \ptt{} takes almost-linear time, by
  \pref{thm:AlmostLinearBottomUpWPO}. After this call the elements and \forward{}s are
  visited only once. Thus, the running time of \ptt{}  is almost-linear. 
\end{proof}

\citet{bourdoncle1993efficient} presents a more efficient version of the
\algwto{} for WTO construction; however, it has worst-case cubic time
complexity. Thus, \pref{thm:AlmostLinearWTO} improves upon the previously known
algorithm for WTO construction.

\begin{figure}[t]
  \begin{minipage}[t]{0.96\textwidth}
\SetInd{.3em}{0.5em}
\begin{algorithm}[H]
  \small
  \SetKw{returnKw}{return}
  \SetKw{stKw}{s.t.}
  \SetKw{thereKw}{there}
  \SetKw{existsKw}{exists}
  \SetKw{andKw}{and}
  \DontPrintSemicolon
  \setlength{\columnsep}{15pt}
  \Indmm
  \KwIn{Directed graph $G(V, \cfgarrow)$, Depth-first forest $D$}
  \KwOut{WTO $\wpo(V, X, \forwardarrow, \backwardarrow)$}
  \Indpp
  \begin{multicols}{2}
  $V, X, \forwardarrow, \backwardarrow \coloneqq \algbu(G, D, \textit{true})$\label{li:ptt-bu}\;
  $stack \coloneqq \emptyset$\;
  \ForEach{$v \in V \cup X$}{
    $count[v] \coloneqq 0$\label{li:ptt-init}\;
    \If{$G.predecessors(v) = \emptyset$}{
      $stack.push(v)$
    }
  }
  $prev \coloneqq \text{nil}$\;
  \While{$stack \neq \emptyset$}{
    $v \coloneqq stack.pop()$\;
    \lIf{$prev \neq \text{nil}$}{
      $\forwardarrow \coloneqq \forwardarrow \cup \{(prev, v)\}$\label{li:ptt-ext}
    }
    $prev \coloneqq v$\;
    \ForEach{$v \forwardarrow w$}{
      $count[w] \coloneqq count[w] + 1$\;
      \lIf{$count[w] = |\wpo.predecessors(w)|$}{\label{li:ptt-top}
        $stack.push(w)$
      }
    }
  }
  \returnKw{$(V, X, \forwardarrow, \backwardarrow)$}
  \end{multicols}
  \caption{\ptt{}$(G, D)$}
  \label{alg:ptt}
  \end{algorithm}
\SetInd{.5em}{1em}  
  \end{minipage}
\vspace{-1ex}
\end{figure}

The next theorem shows that $\ptt{}$ outputs the same WTO as
\algwto{}.
\begin{theorem}
  Given a directed graph $G$ and its depth-first forest $D$, $\ptt{}(G, D)$ and $\algwto{}(G, D)$ construct the same WTO for $G$.
\end{theorem}
\begin{proof}
  We shown above, $\algwto{}(G, D)$ and $\algtd{}(G, D)$ construct the same
components with same heads and nesting relationship. Therefore, using
\pref{thm:same}, we can conclude that $\ptt{}(G, D)$ and $\algwto{}(G, D)$
construct the same WTO.
\end{proof}

The next theorem shows that our concurrent fixpoint algorithm in
\pref{fig:RulesWPO} computes the same fixpoint as Bourdoncle's sequential
fixpoint algorithm.

\begin{theorem}
  \label{thm:precision}
  The fixpoint computed by the concurrent fixpoint algorithm in \pref{fig:RulesWPO}
  using the WPO constructed by \pref{alg:bu} is the same as the one computed by
  the sequential Bourdoncle's algorithm that uses the recursive iteration strategy.
\end{theorem}
\begin{proof}
  With both \backward{}, $\backwardarrow$, and matching parentheses, $(\dots)$,
interpreted as the ``iteration until stabilization'' operator, our concurrent
iteration strategy for $\ptt{}(G, D)$ computes the same fixpoint as Bourdoncle's
recursive iteration strategy for $\algwto{}(G, D)$. The only change we make to a
WPO in \ptt{} is adding more \forward{}s. 
Further, \pref{thm:deter} proved that our concurrent iteration strategy is deterministic.
Thus, our concurrent iteration
strategy computes the same fixpoint
when using the WPO constructed by either $\ptt{}(G, D)$ or $\algbu{}(G, D,
\textit{false})$. Therefore, our concurrent fixpoint algorithm in
\pref{fig:RulesWPO} computes the same fixpoint as Bourdoncle's sequential
fixpoint algorithm.
\end{proof}

\section{Implementation}
\label{sec:implementation}

Our deterministic parallel abstract interpreter, which we called $\pikosname$,
was built using $\ikos$ \cite{ikos2014}, an abstract-interpretation framework
for C/C++ based on LLVM.

\subsubsubsection{Sequential baseline $\ikos$. }
$\ikos$ performs interprocedural analysis to compute invariants for all programs
points, and can detect and prove the absence of runtime errors in programs. To
compute the fixpoint for a function, $\ikos$ constructs the WTO of the CFG of
the function and uses Bourdoncle's recursive iteration
strategy~\cite{bourdoncle1993efficient}. Context sensitivity during
interprocedural analysis is achieved by performing dynamic inlining during
fixpoint: formal and actual parameters are matched, the callee is analyzed, and
the return value at the call site is updated after the callee returns. This
inlining also supports function pointers by resolving the set of possible
callees and joining the results.

\subsubsubsection{$\pikosname$.}
We modified $\ikos$ to implement our deterministic parallel abstract interpreter
using Intel's Threading Building Blocks (TBB) library \cite{TBB}. We
implemented the almost-linear time algorithm for WPO construction
(\pref{sec:Algorithm}).
We implemented the
deterministic parallel fixpoint iterator (\pref{sec:fixpoint}) using TBB's
\texttt{parallel\_do}. Multiple callees at an indirect call site are analyzed in
parallel using TBB's \texttt{parallel\_reduce}. We refer to this extension of
$\ikos$ as $\pikosname$; we use \pikos{k} to refer to the instantiation of
$\pikosname$ that uses up to $k$ threads.

\subsubsubsection{Path-based task spawning in $\pikosname$.} 
$\pikosname$ relies on TBB's \texttt{tasks} to implement the parallel fixpoint
iterator. Our initial implementation would spawn a task for each WPO element when it is
ready to be scheduled. 
Such a naive approach resulted $\pikosname$ being slower than $\ikos$; 
there were 10 benchmarks where speedup of \pikos{2} was below 0.90x compared to $\ikos$, 
with a minimum speedup of 0.74x.
To counter such behavior, we implemented a simple path-based heuristic for
spawning tasks during fixpoint computation. We assign ids to each element in the
WPO $W$ as follows: assign id 1 to the elements along the longest path in $W$,
remove these elements from $W$ and assign id 2 to the elements along the longest path
in the resulting graph, and so on. The length of the path is based on the number
of instructions as well as the size of the functions called along the
path. During the fixpoint computation, a new task is spawned only if the id of
the current element differs from that of the successor that is ready to be
scheduled. Consequently, elements along critical paths are executed in the same
task.

\subsubsubsection{Memory allocator for concurrency.}
We experimented with three memory allocators optimized for parallelism:
Tcmalloc,\footnote{\url{https://gperftools.github.io/gperftools/tcmalloc.html}}
Jemalloc,\footnote{\url{http://jemalloc.net/}}
and Tbbmalloc \cite{TBB}. Tcmalloc
was chosen because it performed the best in our settings for both $\pikosname$
and \ikos.

\subsubsubsection{Abstract domain.}
Our fixpoint computation algorithm is orthogonal to the abstract domain in use.
\pikos{k} works for all abstract domains provided by $\ikos$ as long as the
domain was thread-safe. These abstract domains include interval~\cite{kn:CC77},
congruence~\cite{Granger:IJCM1989}, gauge \cite{venet2012gauge}, and
DBM~\cite{Mine:PADO2001}. Variable-packing domains~\cite{GangeNSSS16:VMCAI2016}
could not be used because their implementations were not thread-safe. We intend
to explore thread-safe implementations for these domains in the future. 

\section{Experimental Evaluation}
\label{sec:evaluation}

In this section, we study the runtime performance of $\pikosname$
(\pref{sec:implementation}) on a large set of C programs using  \ikos{} as the
baseline. 
The experiments were designed to answer the following questions:
\begin{itemize}%
  \item[\RQ{0}] \textbf{[Determinism]} Is $\pikosname$ deterministic?
    Is the fixpoint computed by $\pikosname$ the same as that computed by $\ikos$?
  \item[\RQ{1}] \textbf{[Performance]} How does the performance of \pikos{4} compare to that of \ikos?
  \item[\RQ{2}] \textbf{[Scalability]} How does the performance of \pikos{k} scale
  as we increase the number of threads~$k$?
\end{itemize}

\subsubsubsection{Platform.}
All experiments were run on Amazon EC2 C5, which
use 3.00 GHz Intel Xeon Platinum 8124M CPUs.
\ikos{} and \pikos{k} with $1 \leq k \leq 4$
were run on c5.2xlarge (8 vCPUs, 4 physical cores, 16GB memory),
\pikos{k} with $5 \leq k \leq 8$ on c5.4xlarge (16 vCPUs, 8 physical cores, 32GB memory), and
\pikos{k} with $9 \leq k$ on c5.9xlarge (36 vCPUs, 18 physical cores, 72GB memory).
Dedicated EC2 instances and BenchExec \cite{Beyer2019} were used to improve
reliability of timing results. The Linux kernel version was 4.4, and gcc 8.1.0 was
used to compile \pikos{k} and \ikos.

\subsubsubsection{Abstract Domain.}
We experimented with both interval and gauge domain, and the analysis precision
was set to track immediate values, pointers, and memory.
The results were similar for both interval and gauge domain.
We show the results using the interval domain.
Because we are only concerned with the time taken to perform fixpoint
computation, we disabled program checks, such as buffer-overflow detection, in
both $\ikos$ and $\pikosname$.

\subsubsubsection{Benchmarks.}
We chose \totalboth{} benchmarks from the following two sources:
\begin{itemize}
  \item [\textbf{SVC}]
    We selected all \totalsvc{} benchmarks from the Linux, control-flows, and
    loops categories of SV-COMP 2019 \cite{svcomp}. These categories are well
    suited for numerical analysis, and have been used in recent
    work~\cite{SinghPV:CAV2018,SinghPV:POPL2018}. Programs from these categories
    have indirect function calls with multiple callees at a single call site,
    large switch statements, nested loops, and irreducible CFGs. 
    
  \item [\textbf{OSS}] 
    We selected all \totaloss{} programs from the Arch Linux core packages that
    are primarily written in C and whose LLVM bitcode are obtainable by gllvm.\footnote{\url{https://github.com/SRI-CSL/gllvm}}
    These include, but are not limited to, \texttt{apache}, \texttt{coreutils},
    \texttt{cscope}, \texttt{curl}, \texttt{dhcp}, \texttt{fvwm}, \texttt{gawk},
    \texttt{genius}, \texttt{ghostscript}, \texttt{gnupg}, %
    \texttt{iproute}, \texttt{ncurses}, \texttt{nmap}, \texttt{openssh},
    \texttt{postfix}, \texttt{r}, \texttt{socat}, %
    \texttt{vim}, \texttt{wget}, etc. 
\end{itemize}

We checked that the time taken by $\ikos$ and \pikos{1} was the same; thus, any
speedup achieved by \pikos{k} is due to parallelism in the fixpoint computation.
Note that the time taken for WTO and WPO construction is very small compared to
actual fixpoint computation, which is why \pikos{1} does not outperform \ikos.
The almost-linear algorithm for WPO construction (\pref{sec:AlgorithmBU}) is an
interesting theoretical result, which shows a new connection between
the algorithms of Bourdoncle and Ramalingam. However, the practical impact of
the new algorithm is in preventing stack overflow in the analyzer that occurs
when using a recursive implementation of Bourdoncle's WTO construction
algorithm; see the GitHub issues linked in the footnotes in \pref{sec:introduction}.

There were \totalverylong{} benchmarks for which $\ikos$ took longer than \timelimit{}.
To include these benchmarks, we made the following modification to the dynamic function inliner, which implements
the context sensitivity in interprocedural analysis in both \ikos{} and \pikosname:
if the call depth during the dynamic inlining exceeds the given limit, the
analysis returns $\top$ for that callee.
For each of the \totalverylong{} benchmarks, we determined the largest limit 
for which $\ikos$ terminated within \timelimit.
Because our experiments are designed to understand the performance improvement 
in fixpoint computation, we felt this was a reasonable thing to do. 

\begin{figure}[t]
  \centering
  \includegraphics[height=5cm]{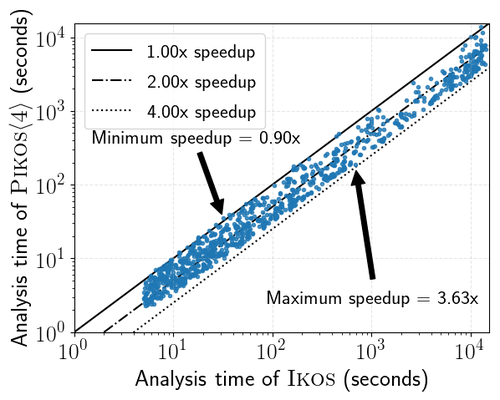}
  \vspace{-1.3ex}
  \caption{Log-log scatter plot of analysis time taken by \ikos{} and \pikos{4}
  on \totalbench{} benchmarks. Speedup is defined as the analysis time of
  \ikos{} divided by analysis time of \pikos{4}. 1.00x, 2.00x, and 4.00x speedup
  lines are shown. Benchmarks that took longer to analyze in $\ikos$ tended to
  have higher speedup.}
  \label{fig:scatplot}
  \vspace{-1ex}
\end{figure}

\begin{table}[t]
  \caption{A sample of the \totalbench{} results in \pref{fig:scatplot}. 
  The first 5 rows list benchmarks with the highest speedup, 
  and the remaining 5 rows list benchmarks with the longest analysis time in $\ikos$.}
  \vspace{-1ex}
  \label{tab:rq1table}
\let\center\empty
\let\endcenter\relax
\centering
\resizebox{0.9\width}{!}{\begin{tabular}{@{}lcp{2cm}p{2cm}p{2cm}@{}}
\toprule
  Benchmark  & Src. &  \ikos{} (s) & \pikos{4} (s) &  Speedup \\
\midrule
        audit-2.8.4/aureport     &       OSS &        684.29 &        188.25 &         3.63x \\
               feh-3.1.3/feh     &       OSS &       9004.83 &       2534.91 &         3.55x \\
ldv-linux-4.2-rc1/43\_2a-crypto  &       SVC &      10051.39 &       2970.02 &         3.38x \\
   ratpoison-1.4.9/ratpoison     &       OSS &       1303.73 &        387.70 &         3.36x \\
ldv-linux-4.2-rc1/08\_1a-gpu-amd &       SVC &       2002.80 &        602.06 &         3.33x \\
\midrule
                     fvwm-2.6.8/FvwmForm &       OSS &      14368.70 &       6913.47 &         2.08x \\
            ldv-linux-4.2-rc1/32\_7a-ata &       SVC &      14138.04 &       7874.58 &         1.80x \\
            ldv-linux-4.2-rc1/43\_2a-ata &       SVC &      14048.39 &       7925.82 &         1.77x \\
   ldv-linux-4.2-rc1/43\_2a-scsi-mpt3sas &       SVC &      14035.69 &       4322.59 &         3.25x \\
ldv-linux-4.2-rc1/08\_1a-staging-rts5208 &       SVC &      13540.72 &       7147.69 &         1.89x \\
\bottomrule
\end{tabular}
}
\vspace{-1ex}
\end{table}
  
\subsection{RQ0: Determinism of $\pikosname$}

As a sanity check for our theoretical results, we experimentally validated \pref{thm:deter} 
by running $\pikosname$ multiple times with varying number of threads,
and checked that the final fixpoint was always the same.
Furthermore, we experimentally validated \pref{thm:precision} by comparing the
fixpoints computed by $\pikosname$ and $\ikos$.

\begin{figure}
  \centering
  \begin{subfigure}[b]{0.45\textwidth}
    \centering
    \includegraphics[width=\textwidth]{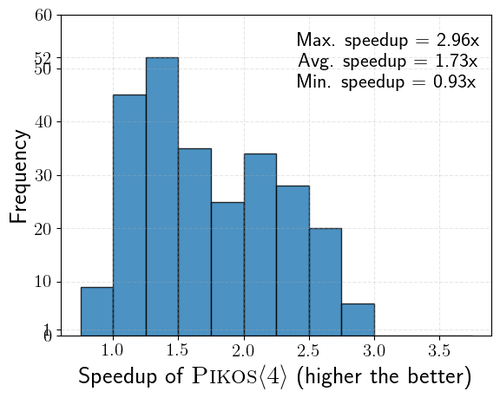}
    \caption[]%
    {{\small 0\% \textasciitilde{} 25\% (5.02 seconds \textasciitilde{} 16.01 seconds)}}
    \label{fig:bottom}
  \end{subfigure}
  \begin{subfigure}[b]{0.45\textwidth}  
    \centering 
    \includegraphics[width=\textwidth]{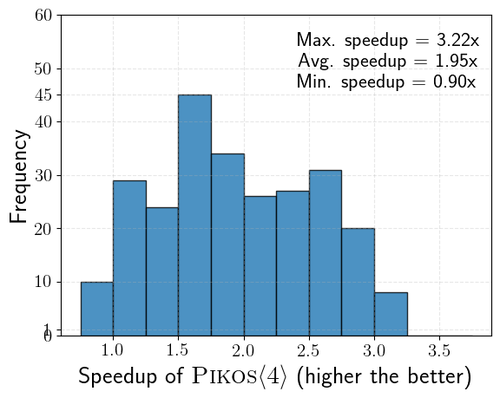}
    \caption[]%
    {{\small 25\% \textasciitilde{} 50\% (16.04 seconds \textasciitilde{} 60.45 seconds)}}
  \end{subfigure}
  \begin{subfigure}[b]{0.45\textwidth}
    \centering
    \includegraphics[width=\textwidth]{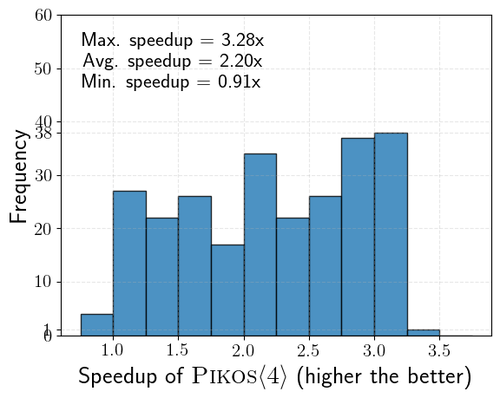}
    \caption[]%
    {{\small 50\% \textasciitilde{} 75\% (60.85 seconds  \textasciitilde{} 508.14 seconds)}}
  \end{subfigure}
  \begin{subfigure}[b]{0.45\textwidth}   
    \centering 
    \includegraphics[width=\textwidth]{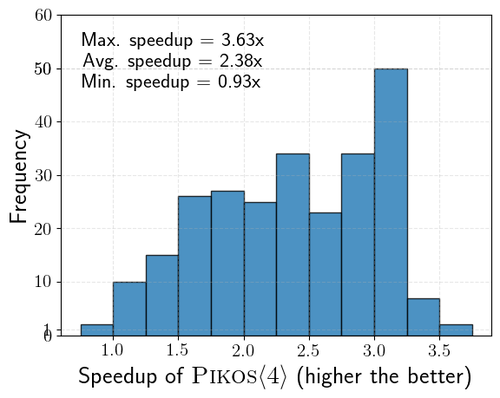}
    \caption[]%
    {{\small 75\% \textasciitilde{} 100\% (508.50 seconds \textasciitilde{} 14368.70 seconds)}}
  \end{subfigure}
  \caption{Histograms of speedup of \pikos{4} for different ranges.
  \pref{fig:bottom} shows the distribution of benchmarks that took from 5.02 seconds
  to 16.01 seconds in \ikos{}. They are the bottom 25\% in terms of the analysis time in \ikos{}.
  The distribution tended toward a higher speedup in the upper range.}
  \label{fig:histplot}
\end{figure}

\subsection{RQ1: Performance of \pikos{4} compared to \ikos}

We exclude results for benchmarks for which $\ikos$ took less than 5 seconds.
For these benchmarks, the average analysis time of \ikos{} was 0.76~seconds. The
analysis time of $\ikos$ minus the analysis time of \pikos{4} ranged from +2.81
seconds (speedup in \pikos{4}) to -0.61 seconds (slowdown in \pikos{4}), with an
average of +0.16 seconds. Excluding these benchmarks left us with
\textbf{\totalbench{}} benchmarks, consisting of \textbf{\totalbenchsvc{} SVC}
and \textbf{\totalbenchoss{} OSS} benchmarks. 

\pref{fig:scatplot} shows a log-log scatter plot, comparing the analysis time of
\pikos{4} with that of $\ikos$ for each of the \totalbench{} benchmarks.
Speedup is defined as the analysis time of \ikos{} divided by
the analysis time of \pikos{4}. 
The maximum speedup was \maxspeed{}, which is close to the maximum speedup of 4.00x.
Arithmetic, geometric, and harmonic mean of the speedup were \ameanspeed{}, \gmeanspeed{}, \hmeanspeed{}, respectively.
Total speedup of all the benchmarks was \totalspeed{}.
As we see in \pref{fig:scatplot}, benchmarks 
for which $\ikos$ took longer to analyze
tended to have greater speedup in \pikos{4}.
Top 25\% benchmarks in terms of the analysis time in \ikos{} had higher averages
than the total benchmarks, with arithmetic, geometric, and harmonic mean of \ameanspeedl{},
\gmeanspeedl{}, and \hmeanspeedl{}, respectively.
\pref{tab:rq1table} shows the results for benchmarks with the highest speedup
and the longest analysis time in \ikos{}.

\pref{fig:histplot} provides details about the distribution of the speedup achieved by \pikos{4}.
Frequency on y-axis represents the number of benchmarks that
have speedups in the bucket on x-axis.
A bucket size of 0.25 is used, ranging from 0.75 to 3.75.
Benchmarks are divided into 4 ranges using the analysis time in \ikos{},
where 0\% represents the benchmark with the shortest analysis time in \ikos{}
and 100\% represents the longest.
The longer the analysis time was in \ikos{} (higher percentile), the more the distribution tended toward a higher speedup for \pikos{4}.
The most frequent bucket was \maxfreqai{} with frequency of
\maxfreqa{} for the range 0\% \textasciitilde{} 25\%.
For the range 25\% \textasciitilde{} 50\%, it was \maxfreqbi{} with frequency of
\maxfreqb{};
for the range 50\% \textasciitilde{} 75\%, \maxfreqci{} with frequency of \maxfreqc{};
and for the range 75\% \textasciitilde{} 100\%, \maxfreqdi{} with frequency of \maxfreqd{}.
Overall, \pikos{4} had speedup over 2.00x for \abovetwonum{} benchmarks out of \totalbench{} (\abovetwo{}), 
and had speedup over 3.00x for \abovethreenum{} out of \totalbench{} (\abovethree{}).

Benchmarks with high speedup contained code with large switch statements nested
inside loops. For example, \texttt{ratpoison-1.4.9/ratpoison}, a tiling window
manager, had an event handling loop that dispatches the events using the switch
statement with 15 cases. Each switch case called an event handler that contained
further branches, leading to more parallelism. Most of the analysis time for
this benchmark was spent in this loop.
On the other hand, benchmarks with low speedup usually had a dominant single
execution path. An example of such a benchmark is \texttt{xlockmore-5.56/xlock},
a program that locks the local X display until a password is entered.

\vspace{0.5em}
\begin{mdframed}[backgroundcolor=light-gray]
  Arithmetic, geometric, and harmonic mean of the speedup exhibited by \pikos{4} were 
  \ameanspeed{}, \gmeanspeed{}, \hmeanspeed{}, respectively.
  Maximum speedup exhibited was \maxspeed{},
  where 4.00x is the maximum possible speedup.
  The performance was generally better for the benchmarks for which $\ikos$ took longer to analyze.
\end{mdframed}

\subsection{RQ2: Scalability of $\pikosname$}

To understand how the performance of \pikos{k} scales with the number of threads $k$,
we carried out the same measurements as in \RQ{1} using \pikos{k} with $k \in \{2,4,6,8\}$.

\begin{figure}[t]
  \centering
    \begin{subfigure}[t]{0.48\textwidth}
      \centering
      \includegraphics[height=4.5cm]{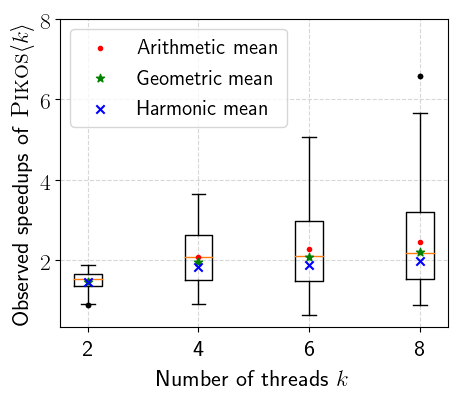}
      \caption{Box plot}
      \label{fig:rq2box}
    \end{subfigure}
    \hfill
    \begin{subfigure}[t]{0.48\textwidth}
      \centering
      \includegraphics[height=4.5cm]{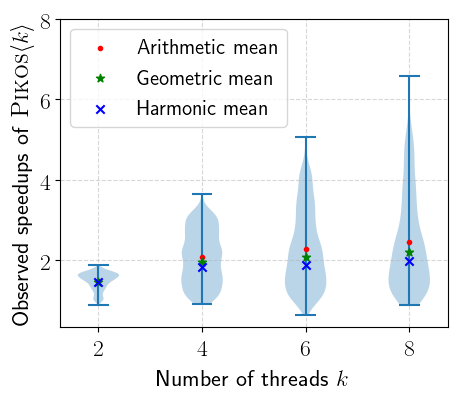}
      \caption{Violin plot}
      \label{fig:rq2violin}
    \end{subfigure}
    \caption{Box and violin plot for speedup of \pikos{k} with $k \in \{2,4,6,8\}$.}
    \label{fig:rq2plot}
\end{figure}

\pref{fig:rq2plot} shows the box and violin plots for speedup obtained by
\pikos{k}, $k \in \{2,4,6,8\}$. 
Box plots show the quartiles and the outliers,
and violin plots show the estimated distribution of the observed speedups.
The box plot \cite{Tukey1977ExploratoryDA} on the
left summarizes the distribution of the results for each $k$ using lower inner
fence ($Q1 - 1.5*(Q3-Q1)$), first quartile~($Q1$), median, third quartile~($Q3$), 
and upper inner fence ($Q3 + 1.5*(Q3-Q1)$). Data beyond the inner fences
(outliers) are plotted as individual points. Box plot revealed that while the
benchmarks above the median~(middle line in the box) scaled, speedups for the
ones below median saturated. Violin plot \cite{violin} on the right supplements
the box plot by plotting the probability density of the results between minimum
and maximum.
In the best case, speedup scaled from \bestscale{}.
For $k \in \{2,4,6,8\}$, the arithmetic means were \ascale{}, respectively;
the geometric means were \gscale{}, respectively; and 
the harmonic means were \hscale{}, respectively.

\begin{figure}
  \centering
    \begin{subfigure}[t]{0.48\textwidth}
      \centering
      \includegraphics[height=4.5cm]{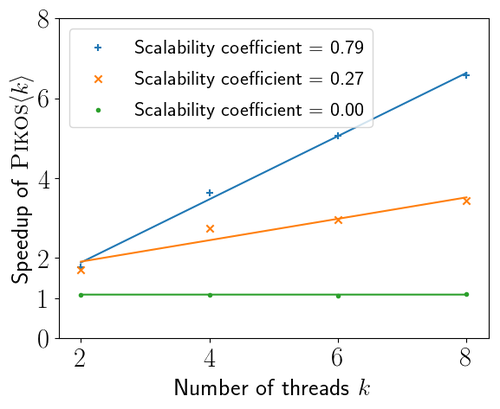}
      \caption{Speedup of \pikos{k} for 3 benchmarks with different scalability coefficients.
      The lines show the linear regressions of these benchmarks.}
      \label{fig:rq2coeff-sample}
    \end{subfigure}
    \hfill
    \begin{subfigure}[t]{0.48\textwidth}
      \centering
      \includegraphics[height=4.5cm]{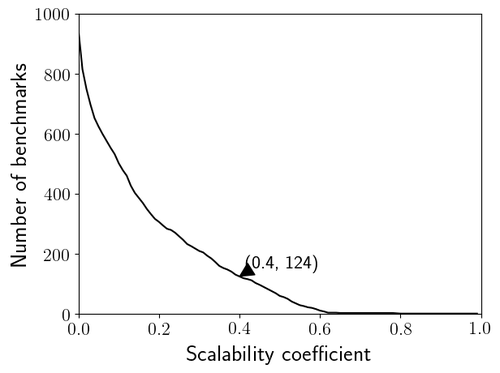}
      \caption{Distribution of scalability coefficients for \totalbench{} benchmarks.
        $(x, y)$ in the plot means that
        $y$ number of benchmarks have scalability coefficient at least $x$.
      }
      \label{fig:rq2coeff-distribution}
    \end{subfigure}
    \caption{Scalability coefficient for \pikos{k}.}
  \label{fig:rq2coeff}
\end{figure}

\begin{table}
  \caption{Five benchmarks with the highest scalability out of \totalbench{} benchmarks.}
  \label{tab:rq2table}
\let\center\empty
\let\endcenter\relax
\centering
\resizebox{0.9\width}{!}{\begin{tabular}{@{}lcrrrrr@{}}
\toprule
                                           &  &   &  \multicolumn{4}{c}{Speedup of \pikos{k}} \\
                                           \cmidrule{4-7}
                                          Benchmark (5 out of \totalbench. Criteria: Scalability) & Src. &  \ikos\ (s) &  $k=4$ &  $k=8$ &  $k=12$ &  $k=16$ \\
\midrule
                           audit-2.8.4/aureport &       OSS &        684.29 &         3.63x &         6.57x &         9.02x &        10.97x \\
                               feh-3.1.3/feh.bc &       OSS &       9004.83 &         3.55x &         6.57x &         8.33x &         9.39x \\
                      ratpoison-1.4.9/ratpoison &       OSS &       1303.73 &         3.36x &         5.65x &         5.69x &         5.85x \\
ldv-linux-4.2-rc1/32\_7a-net-ethernet-intel-igb &       SVC &       1206.27 &         3.10x &         5.44x &         5.71x &         6.46x \\
  ldv-linux-4.2-rc1/08\_1a-net-wireless-mwifiex &       SVC &      10224.21 &         3.12x &         5.35x &         6.20x &         6.64x \\
\bottomrule
\end{tabular}
}
\end{table}

To better measure the scalability of \pikos{k} for individual benchmarks, we
define a \emph{scalability coefficient} as the slope of the linear regression of
the number of threads and the speedups. The maximum scalability coefficient is
1, meaning that the speedup increases linearly with the number of threads. 
If the scalability coefficient is 0, the speedup is the same regardless of the number of
threads used. If it is negative, the speedup goes down with increase in number of
threads. The measured scalability coefficients are shown in \pref{fig:rq2coeff}.
\pref{fig:rq2coeff-sample} illustrates benchmarks exhibiting different
scalability coefficients. For the benchmark with coefficient \bcoeff{}, the speedup of $\pikosname$ roughly
increases by 4, from 2x to 6x, with 6 more threads. For benchmark with coefficient
0, the speedup does not increase with more threads.
\pref{fig:rq2coeff-distribution} shows the distribution of scalability
coefficients for all benchmarks. From this plot we can infer, for instance, that
\coefff{} benchmarks have at least 0.4 scalability coefficient. For these
benchmarks, speedups increased by at least 2 when 5 more threads are given.

\pref{tab:rq2table} shows the speedup of \pikos{k} for $k \geq 4$ for a
selection of five benchmarks that had the highest scalability coefficient in the
prior experiment. In particular, we wanted to explore the limits of scalability
of \pikos{k} for this smaller selection of benchmarks. With scalability
coefficient \bcoeff{}, the speedup of \texttt{audit-2.8.4/aureport} reached
\bspeed{} using 16 threads. This program is a tool that produces summary reports
of the audit system logs. Like \texttt{ratpoison}, it has an event-handler loop
consisting of a large switch statement as shown in \pref{fig:aureport}.

\begin{figure}[t]
  \centering
  \includegraphics[width=13cm]{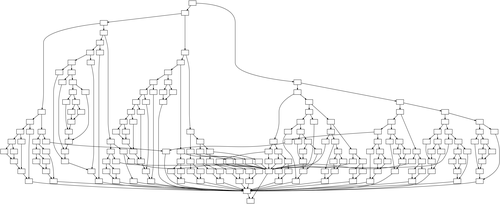}
  \caption{CFG of function \texttt{per\_event\_detailed} in \texttt{aureport} for which
  $\pikosname$ had the maximum scalability. 
  This function calls the appropriate handler based on the event type. This function is called
  inside a loop.}
  \label{fig:aureport}
\end{figure}

\vspace{0.5em}
\begin{mdframed}[backgroundcolor=light-gray]
  The scalability of $\pikosname$ depends on the structure of the analyzed programs.
  $\pikosname$ exhibited high scalability on programs with multiple disjoint paths of similar length.
  For the program \texttt{aureport}, which had the highest scalability, the speedup of \pikos{k} scaled from \bestscaleall{}
  with $k = $ 2, 4, 6, 8, 12, and 16.
\end{mdframed}

\section{Related work}
\label{sec:related}

Since its publication in 1993, Bourdoncle's
algorithm \cite{bourdoncle1993efficient} has become the de facto
approach to solving equations in abstract interpretation. Many
advances have been developed since, but they rely on Bourdoncle's
algorithm; in particular, different ways of intertwining widening and
narrowing during fixpoint computation with an aim to improve
precision \cite{halbwachsSAS2012,amatoSAS2013,amatoSCP2016}.

C Global Surveyor (CGS) \cite{PLDI:VB2004} that performed array bounds checking
was the first attempt at distributed abstract interpretation. It performed
distributed batch processing, and a relational database was used for both
storage and communication between processes. Thus, the communication costs were
too high, and the analysis did not scale beyond four CPUs.

\citet{monniaux2005parallel} describes a parallel
implementation of the ASTR{\'E}E analyzer \cite{cousot2005astree}.
It relies on \emph{dispatch points} that divide the control flow between
long executions, which are found in embedded applications; the
tool analyzes these two control paths in parallel.
Unlike our approach, this parallelization technique is not applicable to programs
with irreducible~CFGs.
The particular parallelization strategy can also lead to a loss in precision. The
experimental evaluation found that the analysis does not scale beyond 4 processors.

\citet{DBLP:conf/cgo/DeweyKH15} present a parallel static analysis for
JavaScript by dividing the analysis into an embarrassingly parallel reachability
computation on a state transition system, and a strategy for selectively merging
states during that reachability computation.

Prior work has explored the use of parallelism for specific program analysis.
BOLT \cite{albarghouthi2012parallelizing} uses a map-reduce framework to
parallelize a top-down analysis for use in verification and software model checking.
Graspan \cite{wangASPLOS2017} implements a single-machine disk-based
graph system to solve graph reachability problems for interprocedural
static analysis. Graspan is not a generic abstract interpreter, and
solves data-flow analyses in the IFDS framework \cite{RHSPOPL1995}.
\citet{su2014parallel} describe a parallel points-to analysis via CFL-reachability.
\citet{FSE:GZL2017} use an actor-model to implement
distributed call-graph analysis. 

\citet{mcpeak2013scalable} parallelize the Coverity
Static Analyzer \cite{CACM:BBC10} to run on an 8-core machine by
mapping each function to its own \emph{work unit}.  
Tricorder \cite{sadowski2015tricorder} is a cloud-based
static-analysis platform used at Google. It supports
only simple, intraprocedural analyses (such as code linters),
and is not designed for distributed whole-program analysis.

Sparse analysis \citep{DBLP:conf/pldi/OhHLLY12,DBLP:journals/toplas/OhHLLPKY14}
and database-backed analysis \citep{WeissRL:ICSE2015} are orthogonal approaches
that improve the memory cost of static analysis.
Newtonian program analysis
\cite{DBLP:journals/toplas/RepsTP17,DBLP:conf/netys/Reps18} provides an
alternative to Kleene iteration used in this paper.

\section{Conclusion}
\label{sec:conclusion}

We presented a generic, parallel, and deterministic algorithm for computing
a fixpoint of an equation system for abstract interpretation. The iteration
strategy used for fixpoint computation is constructed from a weak partial order
(WPO) of the dependency graph of the equation system. We described an axiomatic
and constructive characterization of WPOs, as well as an efficient almost-linear
time algorithm for constructing a WPO. This new notion of WPO generalizes
Bourdoncle's weak topological order (WTO). We presented a linear-time algorithm
to construct a WTO from a WPO, which results in an almost-linear algorithm for
WTO construction given a directed graph. The previously known algorithm for WTO
construction had a worst-case cubic time-complexity. We also showed that the
fixpoint computed using the WPO-based parallel fixpoint algorithm is the same as
that computed using the WTO-based sequential fixpoint algorithm. 

We presented $\pikosname$, our implementation of a WPO-based parallel abstract
interpreter. Using a suite of $\totalbench$ open-source programs and
SV-COMP~2019 benchmarks, we compared the performance of $\pikosname$ against
the $\ikos$ abstract interpreter.
\pikos{4} achieves an average speedup of \ameanspeed{}
over $\ikos$, with a maximum speedup of \maxspeed. 
\pikos{4} showed greater than 2.00x speedup for \abovetwonum{} benchmarks (\abovetwo{})
and greater than 3.00x speedup for \abovethreenum{} benchmarks (\abovethree{}).
\pikos{4} exhibits a larger
speedup when analyzing programs that took longer to analyze using $\ikos$.
\pikos{4} achieved an average speedup of 1.73x on programs for which $\ikos$
took less than 16~seconds, while \pikos{4} achieved an average speedup of
2.38x on programs for which $\ikos$ took greater than 508~seconds. The
scalability of $\pikosname$ depends on the structure of the program being
analyzed with \pikos{16} exhibiting a maximum speedup of \bspeed.

\begin{acks}                            %
The authors would like to thank Maxime Arthaud for help with $\ikos$.
This material is based upon work supported by 
a Facebook Testing and Verification research award, and AWS Cloud Credits for Research.

\end{acks}

\bibliography{main}

\end{document}